\documentclass[11pt]{article}
\usepackage{authblk}
\usepackage[toc,page]{appendix}
\usepackage[top=2cm, bottom=2cm, left=2cm, right=2cm]{geometry}

\usepackage{color}
\usepackage{helvet}         
\usepackage{iftex}
\ifPDFTeX
    \usepackage{courier}    
\fi
\usepackage{type1cm}        

\usepackage{framed} 
\usepackage{tikz}
\usepackage{makeidx}         
\usepackage{graphicx}        
\usepackage{multicol}        
\usepackage[bottom]{footmisc}

\usepackage{amsmath}
\usepackage{amssymb}
\usepackage{bbold}
\usepackage{amsthm}
\usepackage{subcaption}
\usepackage{sidecap}
\usepackage{floatrow}
\usepackage{pdflscape}
\usepackage{comment}
\usepackage[font=small]{caption}
\usepackage{enumitem}
\usepackage{esint}

\usepackage[hidelinks]{hyperref}

\usepackage{scalerel}
\usepackage{shuffle}
\usepackage{mathrsfs}
\usepackage{mathtools}

\mathtoolsset{showonlyrefs,showmanualtags}

\newtheorem{theorem}{Theorem}

\newtheorem{lemma}[theorem]{Lemma}

\newtheorem{proposition}[theorem]{Proposition}

\theoremstyle{definition}

\theoremstyle{remark}

\newtheorem{remark}[theorem]{\bf Remark}
\newtheorem{example}[theorem]{\bf Example}

\numberwithin{theorem}{section}
\numberwithin{figure}{section}
\numberwithin{equation}{section}

\begin{document}
\title{Quasi-exponentials and semi-classical limits for the chordal BPZ system}
\bigskip{}
\author[1]{Mo Chen\thanks{chenmothuprobab@gmail.com. ORCID:0009-0000-6574-4308}}
\author[1]{Hao Wu\thanks{hao.wu.proba@gmail.com. ORCID:0000-0003-4265-7417}}
\affil[1]{Tsinghua University, China}

%
%
\date{}

\newcommand{\Ham}{\mathcal{H}}
\newcommand{\Prad}[1]{\mathsf{P}_{#1\mathrm{\textnormal{-}rad}}}
\newcommand{\Erad}[1]{\mathsf{E}_{#1\mathrm{\textnormal{-}rad}}}
\newcommand{\LZrad}[1]{\mathcal{Z}_{#1\mathrm{\textnormal{-}rad}}}
\newcommand{\LSrad}[1]{\mathsf{S}_{#1\mathrm{\textnormal{-}rad}}}
\newcommand{\LFrad}[1]{\mathcal{F}_{#1\mathrm{\textnormal{-}rad}}}

\newcommand{\Pchord}{\mathsf{P}}
\newcommand{\Echord}{\mathsf{E}}
\newcommand{\LSchord}[1]{\mathsf{S}_{#1\mathrm{\textnormal{-}chord}}}
\newcommand{\QSchord}[1]{\widehat{\mathsf{S}}_{#1\mathrm{\textnormal{-}chord}}}
\newcommand{\LFchord}[1]{\mathcal{F}_{#1\mathrm{\textnormal{-}chord}}}

\newcommand{\blm}{\mathfrak{m}}
\newcommand{\LZmix}{\mathcal{Z}_{\bs{s};\alpha}^{(p; \mu)}}
\newcommand{\LZmixp}{\mathcal{Z}_{\bs{s};\alpha}^{(p; \mu_p)}}
\newcommand{\Cmixp}{C_{\bs{s};\alpha}^{(p; \mu_p)}}
\newcommand{\LZmixpminus}{\mathcal{Z}_{\bs{s};\alpha}^{(p; -\mu_p)}}
\newcommand{\Cmixpminus}{C_{\bs{s};\alpha}^{(p; -\mu_p)}}
\newcommand{\LZmixhat}{\mathcal{Z}_{\bs{s};\hat{\alpha}}^{(p; \mu)}}
\newcommand{\LZalphaCR}{\mathcal{Z}_{\alpha{\textnormal{-}}\mathrm{cr}}^{(\ell;\mu)}}
\newcommand{\CalphaCR}{C_{\alpha{\textnormal{-}}\mathrm{cr}}^{(\ell;\mu)}}
\newcommand{\LZalphaCRhat}{\mathcal{Z}_{\hat{\alpha}{\textnormal{-}}\mathrm{cr}}^{(\ell;\mu)}}
\newcommand{\LZalphahcap}{\mathcal{Z}_{\alpha{\textnormal{-}}\hcap}^{(\lambda)}}
\newcommand{\LZalphahcaphat}{\mathcal{Z}_{\hat{\alpha}{\textnormal{-}}\hcap}^{(\lambda)}}

\newcommand{\Rat}{\operatorname{Rat}}
\newcommand{\QExp}{\operatorname{QExp}}
\newcommand{\QExpr}{\operatorname{QExp}_{\mathfrak{r}}}
\global\long\def\U{\mathbb{U}}
\global\long\def\T{\mathbb{T}}
\global\long\def\HH{\mathbb{H}}
\global\long\def\R{\mathbb{R}}
\global\long\def\C{\mathbb{C}}
\global\long\def\N{\mathbb{N}}
\global\long\def\Z{\mathbb{Z}}
\global\long\def\E{\mathbb{E}}
\global\long\def\PP{\mathbb{P}}
\global\long\def\QQ{\mathbb{Q}}
\global\long\def\A{\mathbb{A}}
\global\long\def\one{\mathbb{1}}

\global\long\def\CR{\mathrm{CR}}
\global\long\def\ST{\mathrm{ST}}
\global\long\def\SF{\mathrm{SF}}
\global\long\def\cov{\mathrm{cov}}
\global\long\def\dist{\mathrm{dist}}
\global\long\def\SLE{\mathrm{SLE}}
\global\long\def\hSLE{\mathrm{hSLE}}
\global\long\def\CLE{\mathrm{CLE}}
\global\long\def\GFF{\mathrm{GFF}}
\global\long\def\inte{\mathrm{int}}
\global\long\def\ext{\mathrm{ext}}
\global\long\def\inrad{\mathrm{inrad}}
\global\long\def\outrad{\mathrm{outrad}}
\global\long\def\dimH{\mathrm{dim}}
\global\long\def\capa{\mathrm{cap}}
\global\long\def\diam{\mathrm{diam}}
\global\long\def\sign{\mathrm{sgn}}
\global\long\def\cat{\mathrm{Cat}}
\global\long\def\cst{\mathrm{C}}
\global\long\def\ck{\mathrm{C}_{\kappa}}
\global\long\def\free{\mathrm{free}}
\global\long\def\hF{{}_2\mathrm{F}_1}
\global\long\def\simple{\mathrm{simple}}
\global\long\def\even{\mathrm{even}}
\global\long\def\odd{\mathrm{odd}}
\global\long\def\st{\mathrm{ST}}
\global\long\def\usf{\mathrm{USF}}
\global\long\def\Leb{\mathrm{Leb}}
\global\long\def\LP{\mathrm{LP}}
\global\long\def\I{\mathrm{I}}
\global\long\def\K{\mathrm{K}}
\global\long\def\J{\mathrm{J}}
\global\long\def\Y{\mathrm{Y}}
\global\long\def\Cat{\mathrm{Cat}}

\global\long\def\II{\mathrm{II}}
\global\long\def\hcap{\mathrm{hcap}}
\global\long\def\const{\mathrm{C}}
\global\long\def\Poisson{\mathrm{H}}
\global\long\def\Green{\mathrm{G}}

\global\long\def\LA{\mathcal{A}}
\global\long\def\LB{\mathcal{B}}
\global\long\def\LC{\mathcal{C}}
\global\long\def\LD{\mathcal{D}}
\global\long\def\LF{\mathcal{F}}
\global\long\def\LK{\mathcal{K}}
\global\long\def\LE{\mathcal{E}}
\global\long\def\LG{\mathcal{G}}
\global\long\def\LI{\mathcal{I}}
\global\long\def\LJ{\mathcal{J}}
\global\long\def\LL{\mathcal{L}}
\global\long\def\LM{\mathcal{M}}
\global\long\def\LN{\mathcal{N}}
\global\long\def\OO{\mathcal{O}}
\global\long\def\LQ{\mathcal{Q}}
\global\long\def\LR{\mathcal{R}}
\global\long\def\LT{\mathcal{T}}
\global\long\def\LS{\mathcal{S}}
\global\long\def\LU{\mathcal{U}}
\global\long\def\LV{\mathcal{V}}
\global\long\def\LW{\mathcal{W}}
\global\long\def\LX{\mathcal{X}}
\global\long\def\LY{\mathcal{Y}}
\global\long\def\PartF{\mathcal{Z}}
\global\long\def\LH{\mathcal{H}}
\global\long\def\LJ{\mathcal{J}}

\global\long\def\blm{\mathfrak{m}}

\global\long\def\LZ{\mathcal{Z}}
\global\long\def\LZrp{\mathcal{Z}_{\alpha; \bs{s}}^{(p)}}
\global\long\def\LJrp{\mathcal{J}_{\alpha; \bs{s}}^{(p)}}
\global\long\def\chamberrp{\chamber_{\alpha; \bs{s}}^{(p)}}
\global\long\def\LErp{\mathcal{E}_{\alpha; \bs{s}}^{(p)}}
\global\long\def\Greenrp{G_{\alpha; \bs{s}}^{(p)}}
\global\long\def\Prp{P_{\alpha; \bs{s}}^{(p)}}
\global\long\def\norcst{\mathrm{C}_{\kappa}^{(\mathfrak{r})}}
\global\long\def\LZalphar{\mathcal{Z}_{\alpha}^{(\mathfrak{r})}}

\global\long\def\coulomb{\LH}
\global\long\def\auxcoulomb{\hat{\coulomb}}
\global\long\def\coulombGas{\LF}
\global\long\def\coulombnew{\LK}
\global\long\def\coulombLine{\LG}
\global\long\def\kfunc{p}

\global\long\def\eps{\epsilon}
\global\long\def\ov{\overline}
\global\long\def\QQrp{\QQ_{\alpha; \bs{s}}^{(p)}}

\global\long\def\bn{\mathbf{n}}
\global\long\def\MR{MR}
\global\long\def\cond{\,|\,}
\global\long\def\la{\langle}
\global\long\def\ra{\rangle}
\global\long\def\tree{\Upsilon}
\global\long\def\prob{\mathbb{P}}
\global\long\def\hm{\mathrm{Hm}}
%

\global\long\def\Im{\operatorname{Im}}
\global\long\def\Re{\operatorname{Re}}

\global\long\def\ud{\mathrm{d}}
\global\long\def\pder#1{\frac{\partial}{\partial#1}}
\global\long\def\pdder#1{\frac{\partial^{2}}{\partial#1^{2}}}
\global\long\def\pddder#1{\frac{\partial^{3}}{\partial#1^{3}}}
\global\long\def\der#1{\frac{\ud}{\ud#1}}

\global\long\def\bZnn{\mathbb{Z}_{\geq 0}}
\global\long\def\bZpos{\mathbb{Z}_{> 0}}
\global\long\def\bZneg{\mathbb{Z}_{< 0}}

\global\long\def\Vfunc{\LG}
\global\long\def\gfunc{g^{(\rr)}}
\global\long\def\hfunc{h^{(\rr)}}

\global\long\def\SimplexInt{\rho}
\global\long\def\CubeInt{\widetilde{\rho}}

\global\long\def\ii{\mathfrak{i}}
\global\long\def\ee{\mathrm{e}}
\global\long\def\rr{\mathfrak{r}}
\global\long\def\chamber{\mathfrak{X}}
\global\long\def\Wchamber{\mathfrak{W}}

\global\long\def\SimplexIntKappa8{\SimplexInt}

\global\long\def\nested{\boldsymbol{\underline{\Cap}}}
\global\long\def\unnested{\boldsymbol{\underline{\cap\cap}}}
\global\long\def\unnested{\boldsymbol{\underline{\cap\cap}}}

\global\long\def\acycle{\vartheta}
\global\long\def\bcycle{\tilde{\acycle}}

\global\long\def\metric{\mathrm{dist}}

\global\long\def\adj#1{\mathrm{adj}(#1)}

\global\long\def\bs{\boldsymbol}

\global\long\def\edge#1#2{\langle #1,#2 \rangle}
\global\long\def\graph{G}

\newcommand{\conn}{\varsigma}
\newcommand{\realacycle}{\smash{\mathring{\acycle}}}
\newcommand{\realpt}{\smash{\mathring{x}}}
\newcommand{\corrind}{\LC}
\newcommand{\bssymb}{\pi}
\newcommand{\PRCM}{\mu}
\newcommand{\coeff}{p}
\newcommand{\MainConst}{C}

\global\long\def\removeLink{/}

\global\long\def\domainofdef{\mathfrak{U}}
\global\long\def\Test_space{C_c^\infty}
\global\long\def\Distr_space{(\Test_space)^*}

\global\long\def\bs{\boldsymbol}
\global\long\def\cst{\mathrm{C}}

\newcommand{\red}{\textcolor{red}}
\newcommand{\blue}{\textcolor{blue}}
\newcommand{\green}{\textcolor{green}}
\newcommand{\magenta}{\textcolor{magenta}}

\newcommand{\coulombGasH}{\mathcal{H}}
\newcommand{\secondbeta}{\intloop}

\newcommand{\cev}[1]{\reflectbox{\ensuremath{\vec{\reflectbox{\ensuremath{#1}}}}}}

\global\long\def\anticonf{\zeta}
\global\long\def\intloop{\varrho}
\global\long\def\Gloop{\smash{\mathring{\intloop}}}

\global\long\def\SLEmeasure{\mathrm{P}}
\global\long\def\SLEmeasureEx{\mathrm{E}}

\global\long\def\fugacity{\nu}
\global\long\def\meanderMat{\mathcal{M}}
\global\long\def\LM{\mathcal{M}}
\global\long\def\meanderMatrix{\meanderMat_{\fugacity}}
\global\long\def\meanderMatrixPrime{\meanderMat_{\fugacity(\kappa')}}
\global\long\def\meanderRenorm{\widehat{\mathcal{M}}}

\global\long\def\PartFRenorm{\widehat{\PartF}}
\global\long\def\coulombGasRenorm{\widehat{\coulombGas}}

\global\long\def\hexa{\scalebox{1.3}{\hexagon}}

\global\long\def\np{p}

\global\long\def\FKdual{\mathcal{L}}

\global\long\def\fixedindex{\flat}
\maketitle
\vspace{-1cm}
\begin{center}
\begin{minipage}{0.95\textwidth}
\abstract{
We study the nonlinear Hamilton--Jacobi system arising as the formal semi-classical limit of the chordal Belavin--Polyakov--Zamolodchikov equations with a common eigenvalue parameter $\lambda$. 
We determine the number of global real-valued solutions modulo additive constants.
When $\lambda=0$, the number of solutions can be derived using the connection between the solutions and rational functions, due to Eremenko.
We focus on the case when $\lambda\neq 0$ in this article.
For $\lambda>0$, a Schwarzian correspondence relates the algebraic equations for the gradients to quasi-exponentials with prescribed critical points.
Enumeration and transversality results of Mukhin, Tarasov, and Varchenko yield smooth real branches, and a Poisson bracket identity establishes their exactness.
An energy bound rules out global real-valued solutions for $\lambda<0$.
}

\end{minipage}
\end{center}

\tableofcontents

\newpage
\section{Introduction}
Commutation relations for Schramm-Loewner evolutions (SLEs) express the requirement that several random curves can be grown in a domain, with their joint law independent of the order of growth.
Dub\'edat~\cite{DubedatCommutationSLE} shows that this compatibility can be encoded by a common partition function $\LZ$, and its logarithmic derivatives determine the driving drifts. Moreover, $\LZ$ satisfies a system of second-order linear partial differential equations.
These are the level-two Belavin--Polyakov--Zamolodchikov (BPZ) equations, which arise in conformal field theory~\cite{BPZInfiniteConformalSymmetry2DQuantum}.
Consider the configuration space
\[\LX_n=\{(x_1, \ldots, x_n)\in\R^n: x_1<\cdots<x_n\}.\]
For $\kappa>0$ and $\lambda\in\R$, we consider the common-eigenvalue version of BPZ equations for $\LZ:\LX_n\to\R$:
\begin{equation}\label{eqn::BPZ_chordal}
\frac{\kappa}{2}\partial_j^2\LZ+\sum_{\ell\ne j}\left(\frac{2}{x_\ell-x_j}\partial_\ell\LZ-\frac{(6-\kappa)/\kappa}{(x_\ell-x_j)^2}\LZ\right)=\frac{\lambda}{\kappa}\LZ ,\qquad \text{for all }j\in\{1, \ldots, n\}.
\end{equation}
When $\lambda=0$, Eq.~\eqref{eqn::BPZ_chordal} is related to partition functions of multiple chordal SLEs; see for example~\cite{BauerBernardKytolaMultipleSLE,GrahamSLE,KytolaPeltolaPurePartitionFunctions,PeltolaWuGlobalMultipleSLEs}.
For every $\lambda\in\R$, the dimension of the solution space of~\eqref{eqn::BPZ_chordal} is always $2^n$, as recently proved in~\cite{ChenHuangWuCommutationChordal}, see more discussion in Appendix~\ref{sec::quantum_vs_semiclassical}.

The limit $\kappa\to0$ connects stochastic SLE growth with deterministic Loewner evolutions.
For multi-chordal SLE, associated with the homogeneous case $\lambda=0$, Peltola and Wang~\cite{PeltolaWangSLELDP} establish a large deviation principle with the multi-chordal Loewner energy as its rate function and identify the limiting curves as minimizers of the corresponding Loewner potential.
At the level of the equations, the semi-classical limit is obtained by the formal substitution $\LZ=\ee^{-\LU/\kappa}$, for which
\[
    \frac{\kappa\partial_j\LZ}{\LZ}=-\partial_j\LU,
    \qquad
    \frac{\kappa^2\partial_j^2\LZ}{\LZ}
    =(\partial_j\LU)^2-\kappa\partial_j^2\LU.
\]
Multiplying~\eqref{eqn::BPZ_chordal} by $\kappa/\LZ$ and retaining the leading-order terms as $\kappa\to0$ yields the following nonlinear first-order system of Hamilton--Jacobi type:
\begin{equation}\label{eqn::BPZ_chordal_sclimits}
\frac{1}{2}(\partial_j\LU)^2-\sum_{\ell\ne j}\left(\frac{2}{x_\ell-x_j}\partial_\ell\LU+\frac{6}{(x_\ell-x_j)^2}\right)=\lambda, \qquad \text{for all }j\in\{1, \ldots, n\}.
\end{equation}
In this article, we consider solutions to~\eqref{eqn::BPZ_chordal_sclimits} and define 
\begin{equation}\label{eqn::solutions_chordal_def}
\LSchord{n}^{(\lambda)}:=\{\LU\in C^{1}(\LX_n\to\R): \LU \text{ satisfies the chordal BPZ system~\eqref{eqn::BPZ_chordal_sclimits}}\}/\sim,
\end{equation}
where $\LU\sim\hat{\LU}$ are equivalent if $\LU-\hat{\LU}$ is a constant. 
The set $\LSchord{n}^{(\lambda)}$ is not empty when $\lambda\ge 0$: define
\begin{align}\label{eqn::LU_shuffle_def}
\LU_{\shuffle_n}^{(\pm)}(\bs{x}):=-2\sum_{1\le j<k\le n}\log(x_k-x_j)\pm\sqrt{2\lambda}\sum_{k=1}^nx_k, \qquad \text{for }\bs{x}=(x_1, \ldots, x_n)\in\LX_n.
\end{align}
One may check by direct calculation that both $\LU_{\shuffle_n}^{(+)}$ and $\LU_{\shuffle_n}^{(-)}$ satisfy~\eqref{eqn::BPZ_chordal_sclimits} with $\lambda\ge 0$. 

For $\lambda\in\R$,
B\'{e}zout's theorem together with the nonintersection of the global solution branches established below gives
\begin{align}\label{eqn::Bezout}
0\le \#\LSchord{n}^{(\lambda)}\le 2^n.
\end{align}
Unlike the dimension of the solution space of~\eqref{eqn::BPZ_chordal}, which is always $2^n$, the number of solutions in $\LSchord{n}^{(\lambda)}$ depends crucially on $\lambda$.
When $\lambda=0$, we have\footnote{The rational Schwarzian correspondence underlying this count is established in~\cite{EremenkoSchwarzian}; see Appendix~\ref{sec::rational} for a proof of~\eqref{eqn::LSchord_zero}, including exactness.}
\begin{align}\label{eqn::LSchord_zero}
\#\LSchord{n}^{(0)}=\binom{n}{\lfloor n/2\rfloor}<2^n. 
\end{align}
In this article, we focus on $\lambda\neq 0$ and show that the upper bound in~\eqref{eqn::Bezout} is achieved when $\lambda>0$ and that the lower bound in~\eqref{eqn::Bezout} is achieved when $\lambda<0$.  
\begin{theorem}\label{thm::BPZ_solutions}
Fix $n\ge 2$. We have 
\begin{align}
\#\LSchord{n}^{(\lambda)}=
\begin{cases}
2^n, \quad&\text{if }\lambda>0; \\
0, \quad &\text{if }\lambda<0. 
\end{cases}
\end{align}
Moreover, when $\lambda>0$, the solutions in $\LSchord{n}^{(\lambda)}$ are all smooth. 
\end{theorem}

For $\lambda=0$ and even $n$, Peltola and Wang~\cite[Proposition~1.8]{PeltolaWangSLELDP} construct solutions in $\LSchord{n}^{(0)}$ from minimal multi-chordal Loewner potentials. Their connection with real rational functions~\cite[Theorem~1.2]{PeltolaWangSLELDP} gives a family indexed by noncrossing pairings, see Remark~\ref{rem::Rat0_PW}. In~\cite[Theorem~2.7]{AlbertsKangMakarovPoleDynamicsMultipleSLE0}, the authors give explicit expressions for these solutions in terms of the poles of the corresponding rational functions. Zhang~\cite[Theorem~3.5 and Section~3.6]{zhang2025multiplechordalsle0classical} extends this construction to general link patterns with unpaired points connected to infinity, yielding additional solutions for $\lambda=0$.
These constructions do not by themselves establish the full solution count for the BPZ system~\eqref{eqn::BPZ_chordal_sclimits} with $\lambda=0$.
The counting problem is explicitly raised in~\cite[Remark~4.6]{PeltolaWangSLELDP}: for $n=2N$, Peltola and Wang predict a total of $\binom{2N}{N}$ solutions by counting planar link patterns with any admissible number of defects, that is, unpaired points connected to infinity.
In~\cite[Section~8]{AlbertsKangMakarovPoleDynamicsMultipleSLE0}, the authors pose a related completeness question with the three conformal Ward identities imposed: for each fixed real configuration of $2N$ points, they expect exactly the $N$-th Catalan number of real solutions to the resulting algebraic system, but leave the matching upper bound open.
These two predicted counts are given by our Theorem~\ref{thm::BPZ_solutions_zero} and Theorem~\ref{thm::third_Ward_sc} respectively.

Our Theorem~\ref{thm::BPZ_solutions} concerns the full solution space.
Together with Theorem~\ref{thm::BPZ_solutions_zero} for $\lambda=0$, this gives a classification of all global real-valued $C^1$ solutions on $\LX_n$, modulo additive constants.
We prove Theorem~\ref{thm::BPZ_solutions} by treating the two signs of $\lambda$ separately. For $\lambda>0$, set $a=\sqrt{\lambda/8}$. The key observation is the connection between the BPZ system~\eqref{eqn::BPZ_chordal_sclimits} and the Schwarzian derivatives of quasi-exponentials that we specify below. Our results for $\lambda>0$ also establish the nonemptiness of the common Hamiltonian level sets studied in~\cite[Proposition~6.5]{AlbertsKangMakarovPoleDynamicsMultipleSLE0} and determine their structure as disjoint unions of $2^n$ global exact graphs; see Section~\ref{subsec::intro_exactness}.

\subsection{Schwarzian derivative of quasi-exponentials}
\paragraph*{Schwarzian derivative.}
The Schwarzian derivative of a nonconstant meromorphic function $f$ is defined as
\begin{equation}\label{eqn::Schwarzian_derivative}
    S_f(z) = \left(\frac{f''(z)}{f'(z)}\right)' - \frac{1}{2}\left(\frac{f''(z)}{f'(z)}\right)^2.
\end{equation}
It is well known that $S_f=0$ if and only if $f$ is a M\"obius transformation.
The Schwarzian derivative has the following chain rule: if $f$ and $g$ are nonconstant meromorphic functions, then
\begin{equation}\label{eqn::Schwarz_chain_rule}
    S_{f\circ g}(z) = S_f(g(z))\cdot (g'(z))^2 + S_g(z).
\end{equation}
In particular, if $L(z)$ is a M\"obius transformation, then $S_{L\circ f}=S_f$.
Conversely, if $S_f=S_g$, then $g=L\circ f$ for some M\"obius transformation $L$.

\paragraph*{Quasi-exponentials.}
We regard two nonconstant meromorphic functions as equivalent if they differ by post-composition with a M\"obius transformation.
A quasi-exponential is such an equivalence class containing a representative of the form
\begin{equation}\label{eqn::quasi_exp_form}
    f(z)=\ee^{2az}\frac{P(z)}{Q(z)},
\end{equation}
where $a\in\R\setminus\{0\}$ and $P,Q\in\C[z]$ are nonzero coprime polynomials. If we require further that $a>0$ and $P$ and $Q$ are monic, then the representative~\eqref{eqn::quasi_exp_form} is unique in its equivalence class. We denote by $f$ both the class and its unique representative~\eqref{eqn::quasi_exp_form}.

The finite critical points of $f$ are encoded by its Wronskian polynomial
\begin{equation}\label{eqn::quasi_exp_Wronskian}
    \mathcal W_a(f)
    :=\frac{1}{2a}\operatorname{Wr}\bigl(\ee^{-az}Q,\ee^{az}P\bigr)
    =\frac{P'Q-PQ'+2aPQ}{2a},
\end{equation}
where $\operatorname{Wr}(g,h)=gh'-g'h$. Indeed,
\begin{equation}\label{eqn::quasi_exp_critical_points}
    f'(z)=2a \ee^{2az}\frac{\mathcal W_a(f)(z)}{Q(z)^2}.
\end{equation}
The multiplicity of a finite critical point $\xi$ is exactly the multiplicity of $\xi$ as a root of $\mathcal W_a(f)$.

Fix $a>0$, $n\geq 1$, and distinct points $z_1,\ldots,z_n\in\C$. We write $\bs z=(z_1,\ldots,z_n)$ and denote by $\QExp(a;\bs{z})$ the set of quasi-exponential equivalence classes with exponential rate $a$ whose finite critical points are precisely $z_1,\ldots,z_n$, all of which are simple. In terms of the unique representative~\eqref{eqn::quasi_exp_form}, we have $f\in\QExp(a;\bs{z})$ if and only if
\begin{equation}\label{eqn::quasi_exp_prescribed_Wronskian}
    \mathcal W_a(f)(z)=\prod_{k=1}^n(z-z_k).
\end{equation}

\begin{proposition}\label{prop::quasi_exp_Schwarz}
Fix $a>0$ and $n\geq 1$. Fix distinct points $z_1,\ldots,z_n\in\C$ and denote $\bs z=(z_1,\ldots,z_n)$.
Suppose $f\in \QExp(a;\bs{z})$.
Then the Schwarzian derivative of $f$ can be written as
\begin{equation}\label{eqn::quasi_exp_Schwarz_explicit}
    S_f(z) =-\frac32\sum_{k=1}^n\frac{1}{(z-z_k)^2} +\frac12\sum_{k=1}^n\frac{u_k}{z-z_k}-2a^2,
\end{equation}
where $(u_1,\ldots,u_n)\in\C^n$ satisfy the following accessory parameter equations:
\begin{align}
    \frac12u_j^2-\sum_{\ell\ne j}\left(\frac{2u_\ell}{z_\ell-z_j}+\frac{6}{(z_\ell-z_j)^2}\right) & =8a^2, \qquad \text{for }j\in\{1,2,\ldots,n\}. \label{eqn::quasi_exp_Schwarz_accessory1}
\end{align}
Conversely, if $(u_1,\ldots,u_n)\in\C^n$ satisfy~\eqref{eqn::quasi_exp_Schwarz_accessory1}, then there exists $f\in \QExp(a;\bs{z})$ whose Schwarzian derivative is given by~\eqref{eqn::quasi_exp_Schwarz_explicit}.
\end{proposition}

The case of $a=0$ of Proposition~\ref{prop::quasi_exp_Schwarz} appears in~\cite{EremenkoSchwarzian}, where the accessory parameters correspond to Schwarzian residues of rational functions. Here we allow a nonzero exponential rate.
This changes the nature of the singularity at infinity in the associated second-order ODE: it is regular singular when $a=0$ and irregular of rank one when $a>0$, and controlling the solutions at this irregular singularity is a key step in the reverse correspondence.
We will prove Proposition~\ref{prop::quasi_exp_Schwarz} in Section~\ref{subsec::quasi_exp_Schwarz}. 
\begin{itemize}
\item The forward implication follows by expanding the Schwarzian derivative at the simple critical points and at infinity; see Lemma~\ref{lem::Schwarz_quasi_exp_expansion}. The accessory equations~\eqref{eqn::quasi_exp_Schwarz_accessory1} are precisely the Frobenius compatibility conditions that exclude logarithmic terms in the associated second-order ODE, as shown in Lemma~\ref{lem::Schwarz_ODE_critical}.
\item The difficult part is the reverse implication, which recovers the global quasi-exponential structure from the BPZ equations themselves. It combines an algebraic step with an analytic one. 
First, we analyze the spectrum of a Lax matrix constructed from the accessory parameters in Lemma~\ref{lem::quasi_exp_Schwarz_accessory_sum}, which forces the condition at infinity~\eqref{eqn::quasi_exp_Schwarz_accessory2}.
Then Lemma~\ref{lem::Schwarz_quasi_exp_expansion_refined_reverse}  reconstructs the quasi-exponential from the associated second-order ODE: multiplication by a local branch of $\sqrt{\prod_{k=1}^n(z-z_k)}$ removes the finite branching, and analysis at infinity shows that the resulting entire solutions can be chosen in the forms $\ee^{az}P(z)$ and $\ee^{-az}Q(z)$.
Their ratio is the required quasi-exponential $\ee^{2az}P/Q$, and Lemma~\ref{lem::Schwarz_ODE} identifies its Schwarzian derivative.
\end{itemize}

Now, let us explain why Proposition~\ref{prop::quasi_exp_Schwarz} is related to the BPZ system~\eqref{eqn::BPZ_chordal_sclimits}. 
On the one hand, Proposition~\ref{prop::quasi_exp_Schwarz} gives a one-to-one correspondence between quasi-exponentials and solutions to~\eqref{eqn::quasi_exp_Schwarz_accessory1}.
On the other hand, when $\bs{z}=\bs{x}\in\LX_n$, the Schubert-calculus results in~\cite[Corollary~7.4(ii)]{MukhinTarasovVarchenkoSpacesQuasiExponentials} give exactly $2^n$ quasi-exponential classes; see Lemma~\ref{lem::QExprax_enumeration} and Proposition~\ref{prop::quasi_exp_enumeration}. The restriction to real critical points is crucial: together with the distinct real exponents $\pm a$, it ensures that the relevant Schubert intersections are transversal. Moreover, the quasi-exponential reality theorem ensures that every class has a representative with real polynomial coefficients, and hence real Schwarzian residues.
Therefore, there are exactly $2^n$ distinct real solutions to~\eqref{eqn::quasi_exp_Schwarz_accessory1} when $\bs{z}=\bs{x}\in\LX_n$. Furthermore, these pointwise solutions form $2^n$ global smooth real branches over the contractible chamber $\LX_n$. Counting the branches alone does not give solutions of the BPZ system: the accessory parameters must also be gradients.
 We show that any such solution $\bs{u}=(u_1, \ldots, u_n)$ is exact in Proposition~\ref{prop::quasi_exp_accessory_exact}. 
Thus, there exists a smooth function $\LU$ such that $\partial_j\LU=u_j$ for $j\in\{1, \ldots, n\}$, and consequently, the system~\eqref{eqn::quasi_exp_Schwarz_accessory1} for $\bs{u}$ is exactly the BPZ system~\eqref{eqn::BPZ_chordal_sclimits} for $\LU$. With $\lambda=8a^2$, this identifies the $2^n$ algebraic branches with the $2^n$ solution classes in Theorem~\ref{thm::BPZ_solutions}.
Proposition~\ref{prop::quasi_exp_LU_poles} further gives an explicit formula for each potential in terms of its quasi-exponential representative.

\subsection{Exactness and Poisson brackets for Hamiltonians}
\label{subsec::intro_exactness}

\begin{proposition}\label{prop::quasi_exp_accessory_exact}
Fix $\lambda>0$ and $n\ge 1$.
Suppose that $\bs{u}=(u_1,\ldots,u_n):\LX_n\to\C^n$ is continuous and, for every $\bs{x}=(x_1,\ldots,x_n)\in\LX_n$, satisfies the accessory parameter equations
\begin{align}\label{eqn::quasi_exp_Schwarz_accessory1_x}
    \frac12u_j^2-\sum_{\ell\ne j}\left(\frac{2u_\ell}{x_\ell-x_j}+\frac{6}{(x_\ell-x_j)^2}\right) & =\lambda, \qquad \text{for }j\in\{1,2,\ldots,n\}.
\end{align}
Then $\bs{u}$ is real-valued and smooth. Moreover, we have
\begin{align}\label{eqn::u_accessory_exact}
\partial_j u_i(\bs{x})=\partial_i u_j(\bs{x}),\qquad\text{for }i\neq j. 
\end{align}
\end{proposition}

When $\lambda=0$, the gradient structure is exhibited by the minimal Loewner potentials of~\cite[Proposition~1.7]{PeltolaWangSLELDP} and the explicit logarithmic potentials of~\cite[Theorem~2.7]{AlbertsKangMakarovPoleDynamicsMultipleSLE0}. 
We will prove Proposition~\ref{prop::quasi_exp_accessory_exact} in Section~\ref{subsec::BPZ_solutions} and our proof is very different. Our proof combines the smooth real branches described above with the Poisson bracket identity~\eqref{eqn::accessory_functional_bracket} which we introduce below.
This establishes exactness for every continuous accessory-parameter branch without using an explicit formula for the potential. The gradient structure is therefore a consequence of the algebraic BPZ equations and their Hamiltonian structure.

\paragraph*{Poisson brackets for Hamiltonians.}
Fix $n\ge 1$. We write $\bs{x}=(x_1, \ldots, x_n)\in\LX_n$ and $\bs{u}=(u_1, \ldots, u_n)\in\R^n$. 
For two smooth functions $G, H: (\bs{x};\bs{u})\in\LX_n\times\R^n\to \R$, we define their Poisson bracket by 
\begin{align}\label{eqn::Poisson_bracket_def}
\{G,H\}:=\sum_{k=1}^n\left(\frac{\partial G}{\partial x_k}\frac{\partial H}{\partial u_k}-\frac{\partial G}{\partial u_k}\frac{\partial H}{\partial x_k}\right).
\end{align}
For $\bs{x}=(x_1,\ldots,x_n)\in\LX_n$ and $\bs{u}=(u_1,\ldots,u_n)\in\R^n$, define the Hamiltonians
\begin{align}\label{eqn::Hamilton_functional}
\Ham_j(\bs{x};\bs{u}):=\frac12u_j^2-\sum_{\ell\ne j}\left(\frac{2u_\ell}{x_\ell-x_j}+\frac{6}{(x_\ell-x_j)^2}\right),\qquad \text{for }j\in\{1,\ldots,n\}.
\end{align}
With the normalization $u_j=\partial_j\LU$, the semi-classical BPZ system~\eqref{eqn::BPZ_chordal_sclimits} is equivalent to
\begin{equation}\label{eqn::BPZ_Hamiltonian}
    \Ham_j\left(\bs{x};\nabla\LU(\bs{x})\right)=\lambda,\qquad \text{for }j\in\{1,\ldots,n\}.
\end{equation}
We will check in Lemma~\ref{lem::quasi_exp_accessory_Poisson} that these Hamiltonians satisfy
\begin{align}\label{eqn::accessory_functional_bracket}
\{\Ham_i, \Ham_j\}+\frac{4(\Ham_i-\Ham_j)}{(x_i-x_j)^2}=0,\qquad\text{for }i\neq j.
\end{align}
This identity is the semi-classical counterpart of the infinitesimal SLE commutation relation~\cite[Section~3.1]{DubedatCommutationSLE}.

In~\cite[Proposition~6.5]{AlbertsKangMakarovPoleDynamicsMultipleSLE0}, the authors study the common level sets of the null-vector Hamiltonians and prove their invariance under the Calogero--Moser flow. However, the nonemptiness and structure of these sets is left undetermined. After a change of momentum variables and normalization, these are the common level sets of the Hamiltonians~\eqref{eqn::Hamilton_functional}. For $\lambda>0$, our results determine their global structure: the set
\[
    \mathcal M_\lambda:=\{(\bs x;\bs u)\in\LX_n\times\R^n:
    \Ham_1(\bs x;\bs u)=\cdots=\Ham_n(\bs x;\bs u)=\lambda\}
\]
is the disjoint union of exactly $2^n$ smooth graphs over the whole chamber $\LX_n$. Proposition~\ref{prop::quasi_exp_accessory_exact} shows that each graph has the form $\bs u=\nabla\LU(\bs x)$ for a globally defined potential.

To prove Proposition~\ref{prop::quasi_exp_accessory_exact}, we first use Propositions~\ref{prop::quasi_exp_Schwarz} and~\ref{prop::quasi_exp_enumeration} to see that a continuous accessory-parameter branch must stay on one of the disjoint smooth real branches. It is therefore itself real-valued and smooth. On the common level set $\Ham_1=\cdots=\Ham_n=\lambda$, the Poisson brackets vanish by~\eqref{eqn::accessory_functional_bracket}. For each fixed $\bs x\in\LX_n$, the $n$ quadratic equations $\Ham_j=\lambda$ have exactly $2^n$ distinct complex solutions, attaining the B\'ezout bound. Hence each solution has multiplicity one, and the Jacobian with respect to $\bs u$ is invertible. Differentiating $\Ham_j(\bs x;\bs u(\bs x))=\lambda$ and combining the resulting identities with the vanishing Poisson brackets shows that $D_{\bs x}\bs u$ is symmetric, which is~\eqref{eqn::u_accessory_exact}. Thus the algebraic equations imply the differential compatibility~\eqref{eqn::u_accessory_exact}.
\medbreak
The Hamiltonian formulation also inspires us to prove the case $\lambda<0$ of Theorem~\ref{thm::BPZ_solutions}.
To separate the kinetic and potential energies in the total Hamiltonian, we shift the momenta by setting
\begin{equation}\label{eqn::u_tilde_def}
\tilde{u}_j:=u_j+\sum_{\ell\ne j}\frac{2}{x_\ell-x_j},\qquad \text{for }j\in\{1,\ldots,n\}.
\end{equation}
A direct calculation gives
\begin{equation}\label{eqn::Hamiltonian_energy}
    \sum_{k=1}^n\Ham_k(\bs{x};\bs{u})
    =\frac12\sum_{k=1}^n\tilde{u}_k^2-\sum_{1\le j<\ell\le n}\frac{16}{(x_\ell-x_j)^2}.
\end{equation}
Since the kinetic energy is nonnegative for real momenta, this identity rules out global real-valued solutions with $\lambda<0$, see Lemma~\ref{lem::BPZ_solutions_neg}.

\paragraph*{Outline.}
We review second-order ODEs and the Schwarzian-ODE correspondence in Section~\ref{sec::pre}.
We analyze quasi-exponentials in Section~\ref{sec::quasi_exp}.
In particular, we prove Proposition~\ref{prop::quasi_exp_Schwarz} in Section~\ref{subsec::quasi_exp_Schwarz}, and prove Proposition~\ref{prop::quasi_exp_accessory_exact} and Theorem~\ref{thm::BPZ_solutions} in Section~\ref{subsec::BPZ_solutions}. 
We prove~\eqref{eqn::LSchord_zero} in Appendix~\ref{sec::rational} using the tools from Sections~\ref{sec::pre} and~\ref{sec::quasi_exp}.
In Appendix~\ref{sec::quantum_vs_semiclassical}, we compare the conclusions for the BPZ system~\eqref{eqn::BPZ_chordal} proved in~\cite{ChenHuangWuCommutationChordal} with our conclusions for the semi-classical limits of the BPZ system~\eqref{eqn::BPZ_chordal_sclimits}.

\paragraph*{Acknowledgment.}
We thank Chongzhi Huang, Eveliina Peltola, and Yilin Wang for helpful discussions.
H.W. is supported by New Cornerstone Investigator Program 100001127. H.W. is partly affiliated at Yanqi Lake Beijing Institute of Mathematical Sciences and Applications, Beijing, China.

\section{Preliminaries}
\label{sec::pre}

We recall the local theory of second-order linear ODEs in Section~\ref{subsec::2ndorder_ODE} and the Schwarzian-ODE correspondence in Section~\ref{subsec::Schwarzian_ODE}. In Section~\ref{subsec::Hamiltonian}, we prove the Poisson bracket identity~\eqref{eqn::accessory_functional_bracket} and the $\lambda<0$ case of Theorem~\ref{thm::BPZ_solutions}.

\subsection{Second-order ODE and its singular points}
\label{subsec::2ndorder_ODE}
Consider the second-order linear ODE
\begin{equation}\label{eqn::second_order_linear_ODE}
    w''(z)+\varphi(z)w'(z)+\psi(z)w(z)=0,
\end{equation}
where $\varphi$ and $\psi$ are holomorphic in a punctured neighborhood of a point $z_0\in \C$. We collect some facts about the ODE~\eqref{eqn::second_order_linear_ODE}; see, e.g.,~\cite[Section~2.7]{OlverLozierBoisvertClarkNIST} for details.
\begin{itemize}
    \item The point $z_0$ is an ordinary point if both $\varphi$ and $\psi$ extend holomorphically to $z_0$; otherwise, it is a singular point.
    \item A singular point $z_0$ is a regular singular point if both $(z-z_0)\varphi(z)$ and $(z-z_0)^2\psi(z)$ extend holomorphically to $z_0$. Otherwise, it is an irregular singular point.
    \item The singularity at $\infty$ is defined by setting $\zeta=1/z$ and $u(\zeta)=w(1/\zeta)$. The equation~\eqref{eqn::second_order_linear_ODE} becomes
    \begin{equation}\label{eqn::second_order_linear_ODE_at_infinity}
        u''(\zeta)+\left(\frac{2}{\zeta}-\frac{\varphi(1/\zeta)}{\zeta^2}\right)u'(\zeta)+\frac{\psi(1/\zeta)}{\zeta^4}u(\zeta)=0.
    \end{equation}
    Thus, if $\infty$ is a singular point, it is regular singular if and only if
    \begin{equation}\label{eqn::regular_singular_at_infinity}
        \varphi(z)=O(z^{-1}) \quad\text{and}\quad \psi(z)=O(z^{-2}) \qquad\text{as }z\to\infty.
    \end{equation}
\end{itemize}

\paragraph*{Regular singular point.}
Assume that $z_0\in \C$ is a regular singular point of~\eqref{eqn::second_order_linear_ODE} and, in a neighborhood of $z_0$, we write
\begin{equation*}
    (z-z_0)\varphi(z)=\sum_{s=0}^{\infty}p_s(z-z_0)^s, \qquad (z-z_0)^2\psi(z)=\sum_{s=0}^{\infty}q_s(z-z_0)^s.
\end{equation*}
We can construct two independent solutions of~\eqref{eqn::second_order_linear_ODE} on a slit neighborhood of $z_0$ of the form
\begin{equation}\label{eqn::Frobenius_solutions}
    \begin{aligned}
    w_1(z)&=(z-z_0)^{r_1}\sum_{s=0}^{\infty}a_s(z-z_0)^s, \qquad \text{with }a_0=1;\\
    w_2(z)&=(z-z_0)^{r_2}\sum_{s=0}^{\infty}b_s(z-z_0)^s + c\,w_1(z)\log(z-z_0),\qquad \text{with }b_0=1;
    \end{aligned}
\end{equation}
where $r_1$ and $r_2$ are the roots of the indicial equation
\begin{equation}\label{eqn::indicial_equation}
    I(r):=r(r-1)+p_0 r+q_0=0.
\end{equation}
If $r_1-r_2\notin\Z$, we can take $c=0$ in the second solution. In the resonant case $r_1-r_2\in\Z$, we need extra compatibility conditions to avoid logarithmic terms. Instead of stating the compatibility conditions in general, we calculate the case needed below.
\begin{lemma}\label{lem::Schwarz_ODE_critical}
Let $z_0\in\mathbb{C}$ and suppose that in a neighborhood of $z_0$, we have
\begin{equation}\label{eqn::simple_critical_ODE_expansion}
    \varphi(z)\equiv 0, \qquad \psi(z)=\sum_{s=0}^{\infty}q_s(z-z_0)^{s-2}, \qquad q_0=-\frac34.
\end{equation}
Then $z_0$ is a regular singular point of~\eqref{eqn::second_order_linear_ODE} and the indicial exponents of~\eqref{eqn::second_order_linear_ODE} at $z_0$ are $3/2$ and $-1/2$. Furthermore,
the equation~\eqref{eqn::second_order_linear_ODE} admits a local Frobenius basis~\eqref{eqn::Frobenius_solutions} without logarithmic terms at $z_0$ if and only if
\begin{equation}\label{eqn::simple_critical_no_log_condition}
    q_1^2+q_2=0.
\end{equation}
\end{lemma}
\begin{proof}
After translating the coordinate, we may assume that $z_0=0$.
Substituting $p_0=0$ and $q_0=-3/4$ into~\eqref{eqn::indicial_equation} gives the indicial exponents $r_1=3/2$ and $r_2=-1/2$.
Their difference is two, thus this is the resonant case of the Frobenius method.
There is always a solution
\[
    w_1(z)=z^{3/2}\sum_{s=0}^{\infty}a_s z^s, \qquad \text{with }a_0=1.
\]
A second log-free Frobenius solution exists precisely when one can solve
\begin{equation}\label{eqn::smaller_exponent_Frobenius_series}
    w_2(z)=z^{-1/2}\sum_{s=0}^{\infty}b_s z^s,
    \qquad \text{with }b_0=1.
\end{equation}
Substituting~\eqref{eqn::smaller_exponent_Frobenius_series} into \eqref{eqn::second_order_linear_ODE} gives, for $s\ge1$,
\begin{equation}\label{eqn::smaller_exponent_Frobenius_recursion}
    I\left(-\frac12+s\right)b_s +\sum_{j=1}^{s}q_jb_{s-j}=0, \qquad I(r):=r(r-1)-\frac34.
\end{equation}
For $s=1$, since $I(1/2)=-1$, this gives $b_1=q_1$.
At the resonant index $s=2$, one has $I(3/2)=0$, thus the compatibility condition is
\begin{align}\label{eqn::Schwarz_ODE_critical_condition}
q_1b_1+q_2b_0=q_1^2+q_2=0.
\end{align}
Let us analyze~\eqref{eqn::Schwarz_ODE_critical_condition}.
\begin{itemize}
\item If the condition~\eqref{eqn::Schwarz_ODE_critical_condition} holds, $b_2$ may be chosen freely and all subsequent coefficients are determined recursively, which gives a convergent log-free solution~\eqref{eqn::smaller_exponent_Frobenius_series}.
Together with $w_1$, this solution forms a local Frobenius basis without logarithmic terms at $z=0$.
\item If the condition~\eqref{eqn::Schwarz_ODE_critical_condition} fails, the second solution is necessarily of the form
\[
    z^{-1/2}\sum_{s=0}^{\infty}b_s z^s + c\,w_1(z)\log z, \qquad \text{with }c\neq 0,
\]
and hence no local Frobenius basis without logarithmic terms exists at $z=0$.
\end{itemize}
\end{proof}

\paragraph*{Irregular singular point.} 
Assume that $\infty$ is an irregular singular point of~\eqref{eqn::second_order_linear_ODE} and, in a neighborhood of $\infty$, we have convergent expansions
\begin{equation}\label{eqn::rank_one_irregular_coefficient_expansions}
    \varphi(z)=\sum_{s=0}^{\infty}\frac{\varphi_s}{z^s},
    \qquad
    \psi(z)=\sum_{s=0}^{\infty}\frac{\psi_s}{z^s},
\end{equation}
with at least one of $\varphi_0,\psi_0,\psi_1$ nonzero.
Assume that the characteristic equation 
\begin{equation}\label{eqn::rank_one_irregular_characteristic_equation}
    \chi^2+\varphi_0 \chi+\psi_0=0
\end{equation}
has distinct roots $\chi_1$ and $\chi_2$ and set 
\begin{equation}\label{eqn::rank_one_irregular_formal_exponents}
    \mu_1=-\frac{\varphi_1 \chi_1+\psi_1}{\varphi_0+2 \chi_1},\qquad  \mu_2=-\frac{\varphi_1 \chi_2+\psi_1}{\varphi_0+2\chi_2}.
\end{equation}
Then~\eqref{eqn::second_order_linear_ODE} has two linearly independent formal solutions of the forms
\begin{equation}\label{eqn::rank_one_irregular_formal_solutions}
    \ee^{\chi_1 z}z^{\mu_1}
    \sum_{s=0}^{\infty}\frac{a_{s}}{z^s},
    \qquad\text{and}\qquad
    \ee^{\chi_2 z}z^{\mu_2}
    \sum_{s=0}^{\infty}\frac{b_{s}}{z^s},
    \qquad\text{with }
    a_{0}=b_0=1.
\end{equation}
These formal series need not converge. Nevertheless, there are linearly independent sectorial solutions $w_1,w_2$ admitting the respective formal expansions~\eqref{eqn::rank_one_irregular_formal_solutions}. That is, there exist sectors $S_1$ and $S_2$ such that for every $M\ge 1$,
\begin{align*}
w_1(z) - \ee^{\chi_1 z} z^{\mu_1} \sum_{s=0}^{M-1} \frac{a_{s}}{z^s} =\ee^{\chi_1 z} z^{\mu_1} O\left(|z|^{-M}\right)
\quad\text{as } z\to\infty \text{ in } S_1;\\
w_2(z) - \ee^{\chi_2 z} z^{\mu_2} \sum_{s=0}^{M-1} \frac{b_{s}}{z^s} =\ee^{\chi_2 z} z^{\mu_2} O\left(|z|^{-M}\right)
\quad\text{as } z\to\infty \text{ in } S_2. 
\end{align*}
For this choice of solutions, their analytic continuations satisfy the Stokes connection formulas
\begin{equation}\label{eqn::rank_one_irregular_Stokes_formulas}
    w_1(z)=\ee^{2\pi\ii\mu_1}w_1(z\ee^{-2\pi\ii})+C_1w_2(z),
    \qquad
    w_2(z)=\ee^{-2\pi\ii\mu_2}w_2(z\ee^{2\pi\ii})+C_2w_1(z),
\end{equation}
where the rotated arguments $z\ee^{\pm2\pi\ii}$ denote analytic continuation and $C_1,C_2$ are called the Stokes multipliers.
\begin{lemma}\label{lem::rank_one_irregular_single_valued}
Suppose that the above sectorial solutions $w_1,w_2$ extend to single-valued solutions in a punctured neighborhood of $\infty$.
Then $\mu_1,\mu_2\in\Z$, the Stokes multipliers $C_1,C_2$ vanish, and the asymptotic expansions~\eqref{eqn::rank_one_irregular_formal_solutions} hold in all directions.
\end{lemma}
\begin{proof}
By single-valuedness,
\[
    w_1(z\ee^{-2\pi\ii})=w_1(z),
    \qquad
    w_2(z\ee^{2\pi\ii})=w_2(z).
\]
Hence, Eq.~\eqref{eqn::rank_one_irregular_Stokes_formulas} becomes
\[
    \bigl(1-\ee^{2\pi\ii\mu_1}\bigr)w_1=C_1w_2,
    \qquad
    \bigl(1-\ee^{-2\pi\ii\mu_2}\bigr)w_2=C_2w_1.
\]
Since $w_1$ and $w_2$ are linearly independent, we obtain
\[
    \ee^{2\pi\ii\mu_1}=\ee^{-2\pi\ii\mu_2}=1,
    \qquad
    C_1=C_2=0.
\]
Thus $\mu_1,\mu_2\in\Z$. The vanishing of both Stokes multipliers means that the asymptotic expansions do not change when the solutions are continued across the Stokes rays, thus they hold in all directions.
\end{proof}

\subsection{Schwarzian-ODE correspondence}
\label{subsec::Schwarzian_ODE}
The ratio of two independent solutions of a second-order linear ODE changes by a M\"obius transformation when the basis of solutions is changed.
The following correspondence therefore allows us to recover a meromorphic function locally, up to M\"obius equivalence, from its Schwarzian derivative by solving a linear ODE.
\begin{lemma}[Schwarzian-ODE correspondence]\label{lem::Schwarz_ODE}
Let $D\subset\C$ be a domain and let $\psi$ be holomorphic on $D$.
If $w_1,w_2$ are linearly independent solutions on $D$ of
\begin{equation}\label{eqn::ODE_general}
    w''(z)+\psi(z)w(z)=0,
\end{equation}
then $f=w_1/w_2$ satisfies
\begin{equation}\label{eqn::Schwarz_ODE}
    S_f(z)=2\psi(z).
\end{equation}
Conversely, if $f$ is a locally univalent meromorphic function on $D$ satisfying~\eqref{eqn::Schwarz_ODE}, then, on any simply connected subdomain on which a branch of $\sqrt{f'}$ is fixed, the two functions
\[
    w_1=\frac{f}{\sqrt{f'}},
    \qquad
    w_2=\frac{1}{\sqrt{f'}}
\]
are linearly independent solutions of~\eqref{eqn::ODE_general}.
\end{lemma}

\begin{proof}
Denote the Wronskian of $w_1$ and $w_2$ by 
\[
    W:=\operatorname{Wr}(w_1,w_2)=w_1w_2'-w_1'w_2.
\]
Since $w_1$ and $w_2$ solve~\eqref{eqn::ODE_general}, we have $W'=0$, and their linear independence implies $W\ne0$. Wherever $w_2\ne0$, we have
\[
    f'=-\frac{W}{w_2^2}, \qquad \frac{f''}{f'}=-2\frac{w_2'}{w_2}.
\]
It follows that
\[
    S_f=-2\frac{w_2''}{w_2}=2\psi.
\]

Conversely, set $w_2=(f')^{-1/2}$ on a simply connected subdomain on which this branch is fixed. Then
\[
    \frac{w_2''}{w_2} =-\frac12\left[\left(\frac{f''}{f'}\right)' -\frac12\left(\frac{f''}{f'}\right)^2\right] =-\frac12S_f=-\psi.
\]
Thus $w_2$ solves~\eqref{eqn::ODE_general}, and $w_1:= f w_2$ solves~\eqref{eqn::ODE_general} as well:
\[
    (w_1)''+\psi w_1 = (fw_2)''+\psi fw_2=f''w_2+2f'w_2'=0.
\]
Finally, $\operatorname{Wr}(w_1,w_2)=-1$, thus $w_1$ and $w_2$ are linearly independent.
\end{proof}

\subsection{Hamiltonian and Poisson bracket}
\label{subsec::Hamiltonian}

We now prove the Poisson bracket identity used to establish exactness in Section~\ref{subsec::BPZ_solutions}, and use the energy decomposition~\eqref{eqn::Hamiltonian_energy} to rule out global solutions with $\lambda<0$.
\begin{lemma}\label{lem::quasi_exp_accessory_Poisson}
Fix $n\ge 1$. Recall that the Poisson bracket is defined in~\eqref{eqn::Poisson_bracket_def} and the Hamiltonians are defined in~\eqref{eqn::Hamilton_functional}. 
Then the Hamiltonians satisfy~\eqref{eqn::accessory_functional_bracket}. 
\end{lemma}

\begin{proof}
Write
\[
    x_{ij}:=\frac{1}{x_i-x_j}, \qquad i\ne j.
\]
Thus $x_{ji}=-x_{ij}$, and, for distinct $i,j,k$,
\begin{equation}\label{eqn::xij_partial_fraction_identity}
    x_{ij}(x_{jk}-x_{ik})=x_{ik}x_{jk}.
\end{equation}
Let us calculate the two sides of~\eqref{eqn::accessory_functional_bracket}. 
\begin{itemize}
\item To compute $\{\Ham_i,\Ham_j\}$, first note that
\begin{equation*}
\frac{\partial \Ham_j}{\partial u_k}
=
\begin{cases}
    u_j,&k=j,\\
    2x_{jk},&k\ne j,
\end{cases}
\qquad
\frac{\partial \Ham_j}{\partial x_k}
=
\begin{cases}
    \displaystyle\sum_{\ell\ne j}
       \left(-2u_\ell x_{j\ell}^2+12x_{j\ell}^3\right),&k=j,\\[6pt]
    2u_k x_{jk}^2-12x_{jk}^3,&k\ne j.
\end{cases}
\end{equation*}
Substituting these expressions into~\eqref{eqn::Poisson_bracket_def}
and grouping the terms with indices $i$, $j$, and
$\ell\notin\{i,j\}$ gives
\begin{align*}
\{\Ham_i,\Ham_j\}
={}&-2x_{ij}^2(u_i^2-u_j^2)-8x_{ij}^3(u_i+u_j)\\
&\quad+x_{ij}\sum_{\ell\notin\{i,j\}}x_{i\ell}x_{j\ell}
  \left(8u_\ell-24(x_{i\ell}+x_{j\ell})\right).
\end{align*}
\item For the Hamiltonian difference, we have
\begin{align*}
\Ham_i-\Ham_j
={}&\frac12(u_i^2-u_j^2)+2(u_i+u_j)x_{ij}+\sum_{\ell\notin\{i,j\}}
  \left(2u_\ell(x_{i\ell}-x_{j\ell})
       -6(x_{i\ell}^2-x_{j\ell}^2)\right)\\
={}&\frac12(u_i^2-u_j^2)+2(u_i+u_j)x_{ij}+\sum_{\ell\notin\{i,j\}}
  \frac{x_{i\ell}x_{j\ell}}{x_{ij}}\left(-2u_\ell
       +6(x_{i\ell}+x_{j\ell})\right).\tag{due to~\eqref{eqn::xij_partial_fraction_identity}}
\end{align*}
\end{itemize}
Combining the two sides, we obtain~\eqref{eqn::accessory_functional_bracket} as desired.
\end{proof}

\begin{lemma}\label{lem::BPZ_solutions_neg}
Fix $n\ge 2$ and $\lambda<0$. Then $\LSchord{n}^{(\lambda)}$ is empty.
\end{lemma}
\begin{proof}
Suppose that $[\LU]\in\LSchord{n}^{(\lambda)}$ and set $\bs{u}(\bs{x})=\nabla\LU(\bs{x})\in\R^n$.
By~\eqref{eqn::BPZ_Hamiltonian} and~\eqref{eqn::Hamiltonian_energy}, the potential energy is bounded from above by the total energy:
\[
    -\sum_{1\le j<\ell\le n}\frac{16}{(x_\ell-x_j)^2}
    \le \sum_{k=1}^n\Ham_k(\bs{x};\bs{u}(\bs{x}))
    =n\lambda,\qquad \bs{x}=(x_1, \ldots, x_n)\in\LX_n.
\]
Fix $\bs{y}\in\LX_n$ and take $\bs{x}=R\bs{y}$ with $R>0$. Then
\[
    -\frac{16}{R^2}\sum_{1\le j<\ell\le n}\frac{1}{(y_\ell-y_j)^2}
    \le n\lambda.
\]
Letting $R\to\infty$ gives $0\le n\lambda$, contradicting $\lambda<0$.
\end{proof}

\section{Quasi-exponentials}
\label{sec::quasi_exp}
The goal of this section is to prove Theorem~\ref{thm::BPZ_solutions} for $\lambda>0$.
We first establish the Schwarzian correspondence of Proposition~\ref{prop::quasi_exp_Schwarz} in Section~\ref{subsec::quasi_exp_Schwarz}.
Section~\ref{subsec::quasi_exp_enumeration} then collects the classification results for quasi-exponentials from~\cite{MukhinTarasovVarchenkoSpacesQuasiExponentials}.
Combining these results, we prove Proposition~\ref{prop::quasi_exp_accessory_exact} and complete the proof of Theorem~\ref{thm::BPZ_solutions} in Section~\ref{subsec::BPZ_solutions}.
Finally, we associate each quasi-exponential $f$ with a relatively explicit solution $\LU_f$ in Proposition~\ref{prop::quasi_exp_LU_poles} in Section~\ref{subsec::pole}.
This formula is not needed for the preceding proofs, but its rational analogue will be used in Appendix~\ref{sec::rational}.

\subsection{Proof of Proposition~\ref{prop::quasi_exp_Schwarz}}
\label{subsec::quasi_exp_Schwarz}

\begin{lemma}\label{lem::Schwarz_quasi_exp_expansion}
Fix $a>0$ and $n\geq 1$. Fix distinct points $z_1,\ldots,z_n\in\C$ and denote $\bs z=(z_1,\ldots,z_n)$.
Suppose $f\in \QExp(a;\bs{z})$.
Then the Schwarzian derivative of $f$ can be written as in~\eqref{eqn::quasi_exp_Schwarz_explicit}
where $(u_1, \ldots, u_n)\in \C^n$ satisfy~\eqref{eqn::quasi_exp_Schwarz_accessory1}.
\end{lemma}

\begin{proof}
Suppose $f(z)=\ee^{2az}P(z)/Q(z)$ where $P, Q$ are nonzero monic coprime polynomials.
Using~\eqref{eqn::quasi_exp_critical_points}, we have
\begin{equation}\label{eqn::quasi_exp_Schwarz_derivative}
    \frac{f''(z)}{f'(z)}
    =2a+\frac{\mathcal W_a(f)'(z)}{\mathcal W_a(f)(z)}-2\frac{Q'(z)}{Q(z)}.
\end{equation}
Thus, $f''/f'$ and $S_f$ are rational.
The expression~\eqref{eqn::quasi_exp_Schwarz_explicit} for $S_f$ follows from the following two observations.
\begin{itemize}
\item The points $z_1, \ldots, z_n$ are double poles of $S_f$ and the coefficient of $(z-z_j)^{-2}$ in $S_f$ is $-3/2$.
After post-composing $f$ with a M\"obius transformation, which does not change its Schwarzian derivative, we may write locally
\[
    f(z)=(z-\xi)^{\tau_{\xi}}\varphi(z), \qquad \varphi(\xi)\ne 0.
\]
It follows that
\[
    \frac{f''(z)}{f'(z)}=\frac{\tau_{\xi}-1}{z-\xi}+H_\xi(z),
\]
where $H_\xi$ is holomorphic near $\xi$.
Substitution into the definition of the Schwarzian derivative gives
\begin{equation}\label{eqn::Schwarz_local_degree_expansion}
    S_f(z)=\frac{1-\tau_{\xi}^2}{2(z-\xi)^2}-\frac{(\tau_{\xi}-1)H_\xi(\xi)}{z-\xi}+O(1).
\end{equation}
Thus $S_f$ is holomorphic at $\xi$ when $\tau_{\xi}=1$, whereas it has a double pole at $\xi$ when $\tau_{\xi}\ge2$.
For $f\in\QExp(a;\bs{z})$, the finite critical points are exactly $z_1,\ldots,z_n$ and each has local degree two.
Hence the coefficient of $(z-z_j)^{-2}$ in $S_f$ is $-3/2$.
\item The rational function $S_f+2a^2$ has no polynomial part.
Using~\eqref{eqn::quasi_exp_prescribed_Wronskian}, we have
\[
    \frac{\mathcal W_a(f)'(z)}{\mathcal W_a(f)(z)}
    =\frac{n}{z}+O(z^{-2}),
    \qquad
    \frac{Q'(z)}{Q(z)}=\frac{\deg(Q)}{z}+O(z^{-2}).
\]
Plugging into~\eqref{eqn::quasi_exp_Schwarz_derivative}, we have
\begin{equation}\label{eqn::Schwarz_expansion_at_infinity}
    S_f(z)=-2a^2-\frac{(n-2\deg(Q))2a}{z}+O(z^{-2})
    \qquad\text{as }z\to\infty.
\end{equation}
In particular, $S_f+2a^2$ has no polynomial part.
\end{itemize}
\medbreak

  It remains to check that $(u_1, \ldots, u_n)\in \C^n$ satisfy~\eqref{eqn::quasi_exp_Schwarz_accessory1}.
 Set
\begin{equation}\label{eqn::accessory_ODE_coefficient}
    \psi_{\bs{u}}(z):=-\frac34\sum_{k=1}^n\frac{1}{(z-z_k)^2}+\frac{1}{4}\sum_{k=1}^n\frac{u_k}{z-z_k}-a^2.
\end{equation}
Thus, Eq.~\eqref{eqn::quasi_exp_Schwarz_explicit} is equivalent to $S_f=2\psi_{\bs{u}}$.
By Lemma~\ref{lem::Schwarz_ODE}, the quasi-exponential $f$ is locally the ratio of two linearly independent solutions of
\[
    w''(z)+\psi_{\bs{u}}(z)w(z)=0.
\]
At $z_j$, the coefficient $\psi_{\bs{u}}$ has the expansion
\[
    \psi_{\bs{u}}(z)
    =-\frac{3}{4(z-z_j)^2}+\frac{u_j}{4(z-z_j)}+q_2^{(j)}
      +O(z-z_j),
\]
where
\begin{equation}\label{eqn::accessory_ODE_coefficient2}
    q_2^{(j)}
    =-a^2-\frac34\sum_{\ell\ne j}\frac{1}{(z_\ell-z_j)^2}
      -\frac{1}{4}\sum_{\ell\ne j}\frac{u_\ell}{z_\ell-z_j}.
\end{equation}
Since $z_j$ is a simple critical point of $f$,
Lemma~\ref{lem::Schwarz_ODE_critical} gives $u_j^2/16+q_2^{(j)}=0$, which is equivalent to~\eqref{eqn::quasi_exp_Schwarz_accessory1} as desired.
\end{proof}

Fix $a>0$ and $n\geq 1$. Fix distinct points $z_1,\ldots,z_n\in\C$ and denote $\bs z=(z_1,\ldots,z_n)$.
Let us classify the quasi-exponentials in $\QExp(a;\bs{z})$ according to different behavior at infinity. 
Such classification will be useful in the proof of Proposition~\ref{prop::quasi_exp_enumeration}. 
Suppose $f(z)=\ee^{2az}P(z)/Q(z)$ where $P, Q$ are nonzero monic coprime polynomials.
Eq.~\eqref{eqn::quasi_exp_Wronskian} and~\eqref{eqn::quasi_exp_prescribed_Wronskian} give
$\deg(P)+\deg(Q)=n$.
Thus, there exists $\mathfrak{r}\in \{0,1,\ldots, n\}$ such that $\deg(P)=n-\mathfrak{r}$ and $\deg(Q)=\mathfrak{r}$. 
For $\mathfrak{r}\in\{0,1,\ldots,n\}$, let $\QExpr(a;\bs{z})$ be the subset of $\QExp(a;\bs{z})$ given by
\begin{equation}\label{eqn::QExpr_def}
\QExpr(a;\bs{z})=\{f(z)=\ee^{2az}P(z)/Q(z)\in\QExp(a;\bs{z}): \deg(P)=n-\mathfrak{r}, \deg(Q)=\mathfrak{r}\}.
\end{equation}
In particular, we have
\begin{equation}\label{eqn::quasi_exp_prescribed_Wronskian_degree}
    \QExp(a;\bs{z})=\bigsqcup_{\mathfrak{r}=0}^n \QExpr(a;\bs{z}).
\end{equation}
For $f\in\QExpr(a;\bs{z})$, we have a refined version of Lemma~\ref{lem::Schwarz_quasi_exp_expansion}. 
\begin{lemma}\label{lem::Schwarz_quasi_exp_expansion_refined}
Fix $a>0$ and $n\geq 1$ and $\mathfrak{r}\in\{0,1,\ldots, n\}$. Fix distinct points $z_1,\ldots,z_n\in\C$ and denote $\bs z=(z_1,\ldots,z_n)$.
Suppose $f\in \QExpr(a;\bs{z})$. 
Then the Schwarzian derivative of $f$ can be written as in~\eqref{eqn::quasi_exp_Schwarz_explicit}
where $(u_1, \ldots, u_n)\in \C^n$ satisfy~\eqref{eqn::quasi_exp_Schwarz_accessory1} and 
\begin{align}\label{eqn::quasi_exp_Schwarz_accessory2}
\frac{1}{4}\sum_{k=1}^n u_k=(2\mathfrak{r}-n)a.
\end{align}
\end{lemma}
\begin{proof}
Expanding the right-hand side of~\eqref{eqn::quasi_exp_Schwarz_explicit} at infinity yields
\begin{equation}\label{eqn::Schwarz_expansion_at_infinity_explicit}
    S_f(z)=-2a^2+\frac{\sum_{k=1}^n u_k}{2z} + O(z^{-2}).
\end{equation}
Comparing this with~\eqref{eqn::Schwarz_expansion_at_infinity} gives~\eqref{eqn::quasi_exp_Schwarz_accessory2} as desired.
\end{proof}

\begin{lemma}\label{lem::quasi_exp_Schwarz_accessory_sum}
Fix $a>0$ and $n\ge 1$. Fix distinct points $z_1, \ldots, z_n\in \C$.
Suppose that $(u_1,\ldots,u_n)\in\C^n$ satisfy~\eqref{eqn::quasi_exp_Schwarz_accessory1}.
Then there exists $\mathfrak{r}\in\{0,1,\ldots, n\}$ such that~\eqref{eqn::quasi_exp_Schwarz_accessory2} holds.
\end{lemma}
\begin{proof}
Define
\begin{equation}\label{eqn::quasi_exp_accessory_parameter_v}
    A_j:=\sum_{\ell\ne j}\frac{1}{z_\ell-z_j}, \qquad v_j:=u_j-2A_j, \qquad \text{for }j\in\{1,\ldots,n\}.
\end{equation}
A direct calculation gives
\[
    v_j^2-4\sum_{\ell\ne j}\frac{v_\ell-v_j}{z_\ell-z_j}
    =u_j^2-4\sum_{\ell\ne j}\frac{u_\ell}{z_\ell-z_j}
      -12\sum_{\ell\ne j}\frac{1}{(z_\ell-z_j)^2}.
\]
Then Eq.~\eqref{eqn::quasi_exp_Schwarz_accessory1} is equivalent to
\begin{equation}\label{eqn::quasi_exp_accessory_parameter_equations_v}
    v_j^2-16a^2=4\sum_{\ell\ne j}\frac{v_\ell-v_j}{z_\ell-z_j}, \qquad \text{for }j\in\{1,\ldots,n\}.
\end{equation}
\medbreak
Let $\bs{v}:=(v_1,\ldots,v_n)^{\mathsf T}$ and introduce the Lax matrix
\begin{equation}
    L_{jk}:=
    \begin{cases}
        4(z_j-z_k)^{-1},&j\ne k,\\
        v_j+4A_j,&j=k.
    \end{cases}
\end{equation}
Writing $\boldsymbol{1}:=(1,\ldots,1)^{\mathsf T}$, we have $L\boldsymbol{1}=\bs{v}$, while \eqref{eqn::quasi_exp_accessory_parameter_equations_v} gives $L\bs{v}=16a^2\boldsymbol{1}$.
Thus
\begin{equation}\label{eqn::quasi_exp_Lax_quadratic_on_one}
    L^2\boldsymbol{1}=16a^2\boldsymbol{1}.
\end{equation}
Set
\[
    Z:=\operatorname{diag}(z_1,\ldots,z_n), \qquad J:=\boldsymbol{1}\boldsymbol{1}^{\mathsf T}.
\]
Then $[Z,L]=4(J-I)$. Hence, for every $r\ge0$,
\[
    0=\frac14\operatorname{tr}([Z,L]L^r)
    =\boldsymbol{1}^{\mathsf T}L^r\boldsymbol{1}-\operatorname{tr}(L^r).
\]
Together with~\eqref{eqn::quasi_exp_Lax_quadratic_on_one}, this gives
\[
    \operatorname{tr}(L^{r+2})=16a^2\operatorname{tr}(L^r),
    \qquad r\ge0.
\]
Let $\chi_1,\ldots,\chi_s$ be the distinct eigenvalues of $L$, with algebraic multiplicities $m_1,\ldots,m_s$.
It follows that
\[
    \sum_{k=1}^s m_k(\chi_k^2-16a^2)\chi_k^r=0,
    \qquad \text{for }r\in\{0,\ldots,s-1\}.
\]
The corresponding Vandermonde matrix is invertible, thus $\chi_k^2=16a^2$ for every $k$.
Therefore, all eigenvalues of $L$ belong to $\{4a,-4a\}$.
Let $\mathfrak r\in\{0,1,\ldots,n\}$ be the algebraic multiplicity of the eigenvalue $4a$.
Since $\sum_{k=1}^n A_k=0$, we obtain
\[
    \frac{1}{4}\sum_{k=1}^n u_k
    =\frac14\operatorname{tr}(L)
    =(2\mathfrak r-n)a,
\]
as desired.
\end{proof}

\begin{lemma}\label{lem::Schwarz_quasi_exp_expansion_refined_reverse}
Fix $a>0$ and $n\geq 1$ and $\mathfrak{r}\in\{0,1,\ldots, n\}$. Fix distinct points $z_1,\ldots,z_n\in\C$ and denote $\bs z=(z_1,\ldots,z_n)$.
Suppose that $(u_1, \ldots, u_n)\in \C^n$ satisfy~\eqref{eqn::quasi_exp_Schwarz_accessory1} and~\eqref{eqn::quasi_exp_Schwarz_accessory2}. Then there exists $f\in \QExpr(a;\bs{z})$ whose Schwarzian derivative is given by~\eqref{eqn::quasi_exp_Schwarz_explicit}.
\end{lemma}

\begin{proof}
Again let $\psi_{\bs{u}}$ be defined by~\eqref{eqn::accessory_ODE_coefficient}.
Consider the equation
\[
    w''(z)+\psi_{\bs{u}}(z)w(z)=0.
\]
At each $z_j$, Eq.~\eqref{eqn::quasi_exp_Schwarz_accessory1} gives $u_j^2/16+q_2^{(j)}=0$, where $q_2^{(j)}$ is defined by~\eqref{eqn::accessory_ODE_coefficient2}.
Lemma~\ref{lem::Schwarz_ODE_critical} therefore shows that, at each $z_j$, the equation admits a local log-free Frobenius basis with exponents $-1/2$ and $3/2$.
Set
\begin{equation}\label{eqn::quasi_exp_R_def}
    R(z):=\prod_{k=1}^n(z-z_k).
\end{equation}
For a local solution $w$, choose a local branch of $\sqrt{R}$ and set the gauged solution $y=\sqrt{R}\,w$.
A direct calculation gives
\begin{equation}\label{eqn::quasi_exp_gauged_ODE}
    y''(z)-\frac{R'(z)}{R(z)}y'(z)
    +\left(-a^2+\frac14\sum_{k=1}^n\frac{v_k}{z-z_k}\right)y(z)=0,
\end{equation}
where $v_j$ is defined by~\eqref{eqn::quasi_exp_accessory_parameter_v}.
The local exponents of~\eqref{eqn::quasi_exp_gauged_ODE} at each $z_j$ are $0$ and $2$, and no logarithmic term occurs.
Consequently, every solution of~\eqref{eqn::quasi_exp_gauged_ODE} extends holomorphically across each $z_j$ and hence to a single-valued entire function.
\medbreak
Next we analyze the singularity of $\infty$ for~\eqref{eqn::quasi_exp_gauged_ODE}.  Eq.~\eqref{eqn::quasi_exp_Schwarz_accessory2} gives $\sum_{k=1}^n v_k=4(2\mathfrak r-n)a$, and hence
\[
    -\frac{R'(z)}{R(z)}=-\frac{n}{z}+O(z^{-2}),
    \qquad
    -a^2+\frac14\sum_{k=1}^n\frac{v_k}{z-z_k}
    =-a^2+\frac{(2\mathfrak{r}-n)a}{z}+O(z^{-2}).
\]
Thus $\infty$ is an irregular singular point for~\eqref{eqn::quasi_exp_gauged_ODE}.
In the notation of~\eqref{eqn::rank_one_irregular_coefficient_expansions},
\[
    \varphi_0=0,
    \qquad
    \varphi_1=-n,
    \qquad
    \psi_0=-a^2,
    \qquad
    \psi_1=(2\mathfrak{r}-n)a.
\]
The roots of the characteristic equation~\eqref{eqn::rank_one_irregular_characteristic_equation} are $\chi_1=a$ and $\chi_2=-a$, and~\eqref{eqn::rank_one_irregular_formal_exponents} gives $\mu_1=n-\mathfrak{r}$ and $\mu_2=\mathfrak{r}$.
Lemma~\ref{lem::rank_one_irregular_single_valued} thus gives linearly independent solutions $y_1,y_2$ satisfying the asymptotic expansions
\[
    y_1(z)\sim\ee^{az}z^{n-\mathfrak{r}}\left(1+\sum_{s\ge1}\frac{a_{s,1}}{z^s}\right),
    \qquad
    y_2(z)\sim\ee^{-az}z^{\mathfrak{r}}\left(1+\sum_{s\ge1}\frac{a_{s,2}}{z^s}\right),\qquad\text{as }z\to\infty\text{ in all directions}.
\]
Consequently,
\[
    P(z):=\ee^{-az}y_1(z),
    \qquad
    Q(z):=\ee^{az}y_2(z)
\]
are entire functions satisfying $P(z)=O(|z|^{n-\mathfrak{r}})$ and $Q(z)=O(|z|^{\mathfrak{r}})$ as $z\to\infty$ in all directions.
The polynomial-growth form of Liouville's theorem shows that $P$ and $Q$ are polynomials of degrees no more than $n-\mathfrak{r}$ and $\mathfrak{r}$, respectively.
Also, the leading terms in the preceding asymptotic expansions give
\[
    \frac{P(z)}{z^{n-\mathfrak{r}}}\longrightarrow 1,
    \qquad
    \frac{Q(z)}{z^{\mathfrak{r}}}\longrightarrow 1, \qquad \text{ as } z\to\infty.
\]
Therefore, $P$ and $Q$ have degrees exactly $n-\mathfrak{r}$ and $\mathfrak{r}$, respectively, and are monic.
\medbreak
Now set
\[
    w_+(z):=\frac{\ee^{az}P(z)}{\sqrt{R(z)}},
    \qquad
    w_-(z):=\frac{\ee^{-az}Q(z)}{\sqrt{R(z)}}.
\]
They are two linearly independent solutions of~\eqref{eqn::ODE_general} with $\psi=\psi_{\bs{u}}$. Since~\eqref{eqn::ODE_general} has no first-derivative term, the Wronskian of $w_-$ and $w_+$ is constant.
By comparing the leading coefficients, we have
\[
    \operatorname{Wr}(w_-,w_+)
    =\frac{P'Q-PQ'+2aPQ}{R}
    =2a.
\]
Since $R$ has only simple zeros, this identity also shows that $P$ and $Q$ are coprime. Therefore,
\[
    \mathcal W_a(f)(z)=R(z),
    \qquad
    f(z):=\frac{w_+(z)}{w_-(z)}
          =\ee^{2az}\frac{P(z)}{Q(z)}\in\QExpr(a;\bs{z}).
\]
Finally, Lemma~\ref{lem::Schwarz_ODE} gives $S_f=2\psi_{\bs{u}}$, which is~\eqref{eqn::quasi_exp_Schwarz_explicit}.
This completes the proof.
\end{proof}
\begin{proof}[Proof of Proposition~\ref{prop::quasi_exp_Schwarz}]
The forward implication is Lemma~\ref{lem::Schwarz_quasi_exp_expansion}.
For the converse, Lemma~\ref{lem::quasi_exp_Schwarz_accessory_sum} supplies the condition at infinity~\eqref{eqn::quasi_exp_Schwarz_accessory2}, and Lemma~\ref{lem::Schwarz_quasi_exp_expansion_refined_reverse} constructs the required quasi-exponential.
\end{proof}

\subsection{Classification of quasi-exponentials}
\label{subsec::quasi_exp_enumeration}
In the following lemma, we collect results on the classification of quasi-exponentials from~\cite{MukhinTarasovVarchenkoSpacesQuasiExponentials}.
We need both an exact count for each prescribed real configuration and smooth real continuation as the critical points vary.
\begin{lemma}\label{lem::QExprax_enumeration}
Fix $a>0$ and $n\geq 1$ and $\mathfrak{r}\in\{0,1,\ldots,n\}$. Fix $\bs{x}\in\LX_n$. The number of quasi-exponentials in $\QExpr(a;\bs{x})$ is $\binom{n}{\mathfrak{r}}$:
\[
    \#\QExpr(a;\bs{x})=\binom{n}{\mathfrak{r}}.
\]
Moreover, for each $f(z)=\ee^{2az}P(z)/Q(z)\in \QExpr(a;\bs{x})$ where $P, Q$ are nonzero monic coprime polynomials, the $n$ non-leading coefficients of $P$ and $Q$ are real and locally smooth in $\bs{x}$.
\end{lemma}
\begin{proof}
This lemma is essentially a restatement of~\cite[Corollary~7.4(ii)]{MukhinTarasovVarchenkoSpacesQuasiExponentials} in our notation. Let $\chi=(\chi_1,\chi_2)$ be the decreasing rearrangement of $(n-\mathfrak{r},\mathfrak{r})$.
We associate $f(z)=\ee^{2az}P(z)/Q(z)\in \QExpr(a;\bs{x})$ with the two-dimensional space of quasi-exponentials
\[
    X_f:=\operatorname{span}_{\C}\bigl\{\ee^{az}P(z),\ee^{-az}Q(z)\bigr\}.
\]
Changing a basis of $X_f$ corresponds to post-composing $f$ with a M\"obius transformation.
The degree condition means that $X_f$ belongs to the affine Schubert cell $\Omega_\chi$ of~\cite[Section~4.1]{MukhinTarasovVarchenkoSpacesQuasiExponentials}.
Moreover, the condition that $x_j$ is a simple critical point of $f$ is expressed by the condition that $X_f$ belongs to the Schubert variety $\Omega_{(1,0)}(x_j)$ of
\cite[Section~5.1]{MukhinTarasovVarchenkoSpacesQuasiExponentials}. Hence
\[
    \QExpr(a;\bs{x})\simeq
    \Omega_{\Lambda,\chi,\bs{x}}
    :=\Omega_\chi\cap\bigcap_{k=1}^n\Omega_{(1,0)}(x_k),
    \qquad \Lambda=((1,0),\ldots,(1,0)).
\]
Since the exponents $-a$ and $a$ are distinct real numbers and the critical points $x_1,\ldots,x_n$ are distinct real numbers, the cited corollary states that this Schubert intersection is transversal and consists of
\[
    \dim\bigl((\C^2)^{\otimes n}\bigr)_\chi
    =\frac{n!}{\chi_1!\chi_2!}
    =\binom{n}{\mathfrak{r}}
\]
distinct points. This proves the counting statement.
\medbreak
We explain how transversality gives local smooth dependence. Let
\[
    \mathscr P_{n}:=z^n+\C[z]_{\leq n-1},
\]
and define the Wronski map
\[
    \pi_{\mathfrak r}:\Omega_\chi\rightarrow \mathscr P_{n},\qquad
    X_f\mapsto\mathcal W_a(f).
\]
Here $\pi_{\mathfrak r}^{-1}(R_{\bs{x}})=\Omega_{\Lambda,\chi,\bs{x}}$, and $R_{\bs{x}}(z)=\prod_{k=1}^n(z-x_k)$ is defined as in~\eqref{eqn::quasi_exp_R_def}.
For fixed $\bs{x}\in\LX_n$, define
\[
    G_{\bs{x}}:\Omega_\chi\rightarrow\C^n,\qquad
    X\mapsto
    \bigl(\pi_{\mathfrak r}(X)(x_1),\ldots,
          \pi_{\mathfrak r}(X)(x_n)\bigr).
\]
The evaluation map
\[
    \operatorname{ev}_{\bs{x}}: T_{R_{\bs{x}}}\mathscr P_{n}
    \rightarrow\C^n,\qquad
    \dot R\mapsto
    \bigl(\dot R(x_1),\ldots,\dot R(x_n)\bigr),
\]
is an isomorphism because the $x_k$ are distinct. Thus transversality means
that
\[
    \ud_X G_{\bs{x}}
    =\operatorname{ev}_{\bs{x}}\circ\,\ud_X\pi_{\mathfrak r}
\]
is invertible, and hence $\ud_X\pi_{\mathfrak r}$ is invertible, at every
$X\in\pi_{\mathfrak r}^{-1}(R_{\bs{x}})$.
The inverse function theorem, together with the smooth dependence of $R_{\bs x}$ on $\bs x$, therefore
gives a local smooth branch
\[
    \bs x\mapsto
    \bigl(\pi_{\mathfrak r}|_V\bigr)^{-1}(R_{\bs x})
\]
through $X$, where $V$ is a neighborhood of $X$ in $\Omega_\chi$.
Since $\Omega_\chi\simeq\C^n$ with the non-leading coefficients of $P$
and $Q$ as coordinates, this proves the asserted local smooth continuation.
\medbreak
Finally we explain the reality statement. It follows from~\cite[Section~7.1]{MukhinTarasovVarchenkoSpacesQuasiExponentials} that every $X_f\in\Omega_{\Lambda,\chi,\bs{x}}$ is real\footnote{The cited corollary~\cite[Corollary~7.4(ii)]{MukhinTarasovVarchenkoSpacesQuasiExponentials} itself does not mention the reality of the spaces of quasi-exponentials. However, this result can indeed be derived from self-adjointness of the Bethe operators as in~\cite[Theorem~7.1]{MukhinTarasovVarchenkoSpacesQuasiExponentials}. See~\cite[Theorem~1.2]{KarpMukhinTarasovPositivityQuasiExponentials} for a clear statement.}. More precisely,
\[
    \sigma(X_f) = X_f,
\]
where $\sigma f(z):=\overline{f(\overline z)}$. In particular, the two subspaces
\[
V_a:=X_f \cap \ee^{az}\C[z]=\C\ee^{az}P(z),\qquad V_{-a}:= X_f \cap \ee^{-az}\C[z]=\C\ee^{-az}Q(z)
\]
are preserved by $\sigma$. Hence
\[
    \overline{P(\overline z)}=c_+P(z),
    \qquad
    \overline{Q(\overline z)}=c_-Q(z)
\]
for some $c_+,c_-\in\C\setminus\{0\}$. Since $P$ and $Q$ are monic,
$c_+=c_-=1$. Therefore $P,Q\in\R[z]$, meaning that the canonical representative $\ee^{2az}P/Q$ of the M\"obius equivalence class has real polynomial coefficients.
\end{proof}

\begin{proposition}\label{prop::quasi_exp_enumeration}
Fix $a>0$ and $n\ge 1$. 
Fix $\bs{x}=(x_1, \ldots, x_n)\in\LX_n$. The number of quasi-exponentials in $\QExp(a;\bs{x})$ is $2^n$:
\begin{align}\label{eqn::QExpax_enumeration}
\#\QExp(a;\bs{x})= 2^n.
\end{align}
Moreover, as $\bs x$ varies over $\LX_n$, these quasi-exponentials form $2^n$ global branches. Along each branch, the $n$ non-leading coefficients of the nonzero monic coprime polynomials $P$ and $Q$ in $f(z)=\ee^{2az}P(z)/Q(z)$ are real-valued smooth functions of $\bs x$.
\end{proposition}

\begin{proof}
Recall that $\QExp(a;\bs{x})=\bigsqcup_{\mathfrak{r}=0}^n \QExpr(a;\bs{x})$.
Combining with Lemma~\ref{lem::QExprax_enumeration}, we have
\[\#\QExp(a;\bs{x})=\sum_{\mathfrak{r}=0}^n \binom{n}{\mathfrak{r}}=2^n,\]
which proves~\eqref{eqn::QExpax_enumeration} as desired.
\medbreak
It remains to globalize the local smooth branches. Fix
$\mathfrak r\in\{0,1,\ldots,n\}$, let $\chi=(\chi_1,\chi_2)$ be the
decreasing rearrangement of $(\mathfrak r,n-\mathfrak r)$, and retain the
notation $\Omega_\chi$ and $\pi_{\mathfrak r}$ from the proof of
Lemma~\ref{lem::QExprax_enumeration}.
Let
\[
    \mathcal E_{\mathfrak r}
    :=\{(\bs x,X)\in\LX_n\times\Omega_\chi:
          \pi_{\mathfrak r}(X)=R_{\bs x}\},
\]
and let
\[
    p_{\mathfrak r}:\mathcal E_{\mathfrak r}\longrightarrow\LX_n,
    \qquad
    (\bs x,X)\longmapsto\bs x.
\]
The map $p_{\mathfrak r}$ is a local smooth diffeomorphism, and every fiber contains exactly $\binom{n}{\mathfrak r}$ points.
Therefore $p_{\mathfrak r}$ is a $\binom{n}{\mathfrak r}$-sheeted covering.
Since $\LX_n$ is contractible, this covering is trivial.
Thus we obtain $\binom{n}{\mathfrak r}$ global sections of $p_{\mathfrak r}$.
These sections are smooth because they are locally given by the smooth inverse branches of $p_{\mathfrak r}$.
Finally, since $\Omega_\chi\simeq\C^n$ with the non-leading coefficients
of $P$ and $Q$ as coordinates, these coefficients are smooth functions of
$\bs{x}\in\LX_n$ along each global section.
\end{proof}

\subsection{Proof of Proposition~\ref{prop::quasi_exp_accessory_exact} and Theorem~\ref{thm::BPZ_solutions}}
\label{subsec::BPZ_solutions}

\begin{proof}[Proof of Proposition~\ref{prop::quasi_exp_accessory_exact}]
Fix $\lambda>0$ and $n\ge 1$.
Let $\bs{u}=(u_1,\ldots,u_n):\LX_n\to\C^n$ be continuous and satisfy the accessory parameter equations~\eqref{eqn::quasi_exp_Schwarz_accessory1_x} at every $\bs{x}\in\LX_n$:
\begin{align*}
    \frac12u_j^2-\sum_{\ell\ne j}\left(\frac{2u_\ell}{x_\ell-x_j}+\frac{6}{(x_\ell-x_j)^2}\right) & =\lambda, \qquad \text{for }j\in\{1,\ldots,n\}.
\end{align*}
By Propositions~\ref{prop::quasi_exp_Schwarz} and~\ref{prop::quasi_exp_enumeration}, as $\bs x$ varies over $\LX_n$, the solution set of~\eqref{eqn::quasi_exp_Schwarz_accessory1_x} is the disjoint union of $2^n$ global smooth real-valued sheets. Indeed, the accessory parameters $u_j$ in~\eqref{eqn::quasi_exp_Schwarz_explicit} depend smoothly on the coefficients of the corresponding quasi-exponential, and they are real because that quasi-exponential has a representative with real polynomial coefficients. Since $\LX_n$ is connected, a continuous section cannot switch between the disjoint sheets. Thus $\bs u$ lies on one fixed sheet and is real-valued and smooth.
\medbreak
For
$(\bs{x},\bs{u})\in\LX_n\times\R^n$ and
$j\in\{1,\ldots,n\}$, define $\Ham_j(\bs{x};\bs{u})$ as in~\eqref{eqn::Hamilton_functional} and we write $\bs{\Ham}=(\Ham_1,\ldots,\Ham_n)$.
Differentiating
$\bs{\Ham}(\bs x;\bs{u}(\bs x))=\lambda\boldsymbol{1}$ gives the Jacobian identity
\begin{equation}\label{eqn::accessory_differentiated_matrix}
    D_{\bs x}\bs \Ham
    +D_{\bs u}\bs \Ham\,D_{\bs x}\bs u=0,
\end{equation}
where all the Jacobians of $\bs \Ham$ are evaluated at
$(\bs x;\bs u(\bs x))$.
We first note that $D_{\bs u}\bs \Ham$ is invertible along the chosen branch.
Indeed, for fixed $\bs x\in\LX_n$, regard $\bs \Ham(\bs x; \cdot)-\lambda\mathbf{1}$ as a polynomial map on $\C^n$.
Proposition~\ref{prop::quasi_exp_enumeration} shows that its zero set consists of exactly $2^n$ distinct points.
B\'ezout's theorem therefore shows that every zero has multiplicity one, or equivalently,
$D_{\bs u}\bs \Ham$ is invertible at every zero.
\medbreak
Next, let us show~\eqref{eqn::u_accessory_exact}: $\partial_j u_i(\bs{x})=\partial_i u_j(\bs{x})$.
From the definition of the Poisson bracket in~\eqref{eqn::Poisson_bracket_def}, we have
\[
    \left(
       D_{\bs x}\bs \Ham(D_{\bs u}\bs \Ham)^{\mathsf T}
       -D_{\bs u}\bs \Ham(D_{\bs x}\bs \Ham)^{\mathsf T}
    \right)_{ij}
    =\{\Ham_i,\Ham_j\}.
\]
Eq.~\eqref{eqn::accessory_functional_bracket} and the identities $\Ham_j=\lambda$
along the branch therefore give
\[
    D_{\bs x}\bs \Ham(D_{\bs u}\bs \Ham)^{\mathsf T}
    =D_{\bs u}\bs \Ham(D_{\bs x}\bs \Ham)^{\mathsf T}.
\]
Together with~\eqref{eqn::accessory_differentiated_matrix}, we have
\[
    D_{\bs x}\bs u
    =-(D_{\bs u}\bs \Ham)^{-1}D_{\bs x}\bs \Ham
    =-\left((D_{\bs u}\bs \Ham)^{-1}D_{\bs x}\bs \Ham\right)^{\mathsf T}
    =(D_{\bs x}\bs u)^{\mathsf T},
\]
which is~\eqref{eqn::u_accessory_exact} as desired. 
\end{proof}

\begin{proof}[Proof of Theorem~\ref{thm::BPZ_solutions}]
The conclusion for $\lambda<0$ is proved in Lemma~\ref{lem::BPZ_solutions_neg}. 
In the following, we assume $\lambda>0$ and set $a=\sqrt{\lambda/8}>0$.
For $\bs{x}\in\LX_n$,
Proposition~\ref{prop::quasi_exp_Schwarz}
and Proposition~\ref{prop::quasi_exp_enumeration}
give $2^n$ distinct global smooth real-valued
branches $\bs u=(u_1,\ldots,u_n)$ satisfying~\eqref{eqn::quasi_exp_Schwarz_accessory1_x}.
For each such $\bs u=(u_1,\ldots,u_n)$,
Proposition~\ref{prop::quasi_exp_accessory_exact} implies that $\partial_j u_i=\partial_i u_j$.
Since $\LX_n$ is contractible, there exists a real-valued smooth function $\LU$, unique up to an additive
constant, such that $\partial_j\LU=u_j$ for all $j\in\{1, \ldots, n\}$. Moreover, Eq.~\eqref{eqn::quasi_exp_Schwarz_accessory1_x} gives
\begin{equation}\label{eqn::BPZ_sclimits_asquare}
\frac12(\partial_j\LU)^2
-\sum_{\ell\ne j}\left(
    \frac{2\partial_\ell\LU}{x_\ell-x_j}
    +\frac{6}{(x_\ell-x_j)^2}
\right)
=8a^2,
\qquad j\in\{1,\ldots,n\},
\end{equation}
which is the same as the BPZ system~\eqref{eqn::BPZ_chordal_sclimits} with $\lambda=8a^2>0$.
Therefore, each branch $\bs{u}=(u_1, \ldots, u_n)$ gives a solution
in $\LSchord{n}^{(\lambda)}$.
Conversely, let $[\LU]\in\LSchord{n}^{(\lambda)}$ and set
\[
\bs u:=(\partial_1\LU,\ldots,\partial_n\LU).
\]
Because $\LU$ is $C^1$, the map $\bs u$ is continuous, and~\eqref{eqn::BPZ_sclimits_asquare} shows that it satisfies~\eqref{eqn::quasi_exp_Schwarz_accessory1_x}. Proposition~\ref{prop::quasi_exp_accessory_exact} therefore places $\bs u$ on one of the $2^n$ smooth real-valued global branches. Thus, there is one-to-one correspondence between solutions to~\eqref{eqn::quasi_exp_Schwarz_accessory1_x} and solutions in $\LSchord{n}^{(\lambda)}$, and we have $\#\LSchord{n}^{(\lambda)}=2^n$. Every solution $\LU$ is smooth, because all its first-order partial derivatives are smooth.
\end{proof}

\subsection{Expansion at poles}
\label{subsec::pole}

\begin{proposition}\label{prop::quasi_exp_LU_poles}
Fix $\lambda>0$ and $n\geq 1$. 
Set $a=\sqrt{\lambda/8}>0$. 
For $\bs{x} = (x_1,\ldots, x_n)\in \LX_n$, let
$f(z)=\ee^{2az}P(z)/Q(z)\in\QExp(a;\bs{x})$ where $P,Q$ are
nonzero monic coprime polynomials. 
Assume $\mathfrak{r}=\deg(Q)\in\{0,1,\ldots,n\}$ and denote the zeros\footnote{When $\mathfrak{r}=0$, we have $Q=1$ and the list of zeros is empty; all sums over empty set are understood to be 0 and all products over empty index sets are understood to be $1$.} of $Q$ by $\zeta_1,\ldots,\zeta_{\mathfrak{r}}$, counted with multiplicity. Define 
\begin{align}\label{eqn::LU_f_def}
\LU_f(\bs{x})=-2\sum_{1\le j<k\le n}\log(x_k-x_j)-4a\sum_{k=1}^n x_k
         +4\sum_{m=1}^{\mathfrak{r}}\left(\log|P(\zeta_m)|+2a\zeta_m\right).
\end{align}
Then $\LU_f$ is the solution in $\LSchord{n}^{(\lambda)}$ corresponding to $f$. 
\end{proposition}

The simplest case in Proposition~\ref{prop::quasi_exp_LU_poles} is $Q=1$, see Example~\ref{example::LUminus}. 
To prove Proposition~\ref{prop::quasi_exp_LU_poles}, we first show the conclusion when $Q$ has only simple zeros in Lemma~\ref{lem::quasi_exp_Schwarz_accessory_determination}. The general case can be obtained by taking continuous extension.
Our formulas extend those of~\cite[Theorem~2.7]{AlbertsKangMakarovPoleDynamicsMultipleSLE0}, where the corresponding result is established for a special family of solutions with $a=0$. The general case for $a=0$ is treated in Appendix~\ref{subsec::rational_poles}, see also~\cite{zhang2025multiplechordalsle0classical}.

\begin{example}\label{example::LUminus}
Fix $\lambda>0$ and $n\geq 1$. 
Set $a=\sqrt{\lambda/8}>0$. The solution $\LU_{\shuffle_n}^{(-)}$ defined in~\eqref{eqn::LU_shuffle_def}: 
\[\LU_{\shuffle_n}^{(-)}(\bs{x})=-2\sum_{1\le j<k\le n}\log(x_k-x_j)-4a\sum_{k=1}^n x_k\]
corresponds to
the quasi-exponential $f(z)=\ee^{2az}P(z)$, where $P$ is the unique monic polynomial of degree $n$ satisfying
\[
    P'(z)+2aP(z)=2a\prod_{k=1}^n(z-x_k).
\]
\end{example}

\begin{lemma}\label{lem::quasi_exp_Schwarz_accessory_determination}
Assume the same setup as in Proposition~\ref{prop::quasi_exp_LU_poles}. 
Assume in addition that $Q$ has only
simple zeros $\zeta_1,\ldots,\zeta_{\mathfrak{r}}$, none of which
belongs to $\{x_1,\ldots,x_n\}$.
\begin{itemize}
\item The poles $\zeta_1,\ldots,\zeta_{\mathfrak{r}}$ satisfy the stationary relations
\begin{equation}\label{eqn::quasi_exp_pole_equations}
    2a+\sum_{k=1}^n\frac{1}{\zeta_j-x_k}
       -2\sum_{m\ne j}
          \frac{1}{\zeta_j-\zeta_m}=0,
    \qquad \text{for }j\in\{1,\ldots,\mathfrak{r}\}.
\end{equation}
\item The accessory parameters $(u_1,\ldots,u_n)$ of $f$ in~\eqref{eqn::quasi_exp_Schwarz_explicit} can also be written as
\begin{equation}\label{eqn::quasi_exp_accessory_pole_expansion}
    u_j=-4a+2\sum_{\ell\neq j}\frac{1}{x_\ell-x_j}
          -4\sum_{m=1}^{\mathfrak{r}}
             \frac{1}{\zeta_m-x_j},
    \qquad \text{for }j\in\{1,\ldots,n\}.
\end{equation}
\item Define  
\begin{equation}\label{eqn::quasi_exp_pole_potential}
\begin{aligned}
    \widehat{\LU}_f(\bs{x})
    =&-2\sum_{1\le j<k\le n}\log(x_k-x_j)-4a\sum_{k=1}^n x_k
         \\
       &+4\sum_{k=1}^n\sum_{m=1}^{\mathfrak{r}}\log|x_k-\zeta_m|
         -8\sum_{1\leq\ell<m\leq \mathfrak{r}}\log|\zeta_\ell-\zeta_m|
         +8a\sum_{m=1}^{\mathfrak{r}}\zeta_m.
\end{aligned}
\end{equation}
Then $\partial_j\widehat{\LU}_f(\bs{x})=u_j$ for all $j\in\{1, \ldots, n\}$.
\end{itemize}
\end{lemma}
\begin{proof}
We first show~\eqref{eqn::quasi_exp_pole_equations}.
From~\eqref{eqn::quasi_exp_critical_points}, we have
\[
    f'(z)=2a\ee^{2az}
       \frac{\prod_{k=1}^n(z-x_k)}
            {\prod_{m=1}^{\mathfrak{r}}(z-\zeta_m)^2}.
\]
A direct calculation gives
\begin{align*}
    \operatorname{Res}_{z=\zeta_j}f'(z)
    &=\lim_{z\to \zeta_j}\left(2a\ee^{2az}
       \frac{\prod_{k=1}^n(z-x_k)}
            {\prod_{m\neq j}(z-\zeta_m)^2}\right)'\\
    &=
      2a\ee^{2a\zeta_j}
       \underbrace{\frac{\prod_{k=1}^n(\zeta_j-x_k)}
            {\prod_{m\neq j}(\zeta_j-\zeta_m)^2}}_{\neq 0}\left(
        2a+\sum_{k=1}^n\frac{1}{\zeta_j-x_k}
           -2\sum_{m\neq j}
              \frac{1}{\zeta_j-\zeta_m}
      \right).
\end{align*}
Since the derivative of a meromorphic function has zero residue at every pole, we have $\operatorname{Res}_{z=\zeta_j}f'(z)=0$. This proves~\eqref{eqn::quasi_exp_pole_equations}.
\medbreak
Next, we derive~\eqref{eqn::quasi_exp_accessory_pole_expansion}. Eq.~\eqref{eqn::quasi_exp_Schwarz_derivative} gives
\[
    \frac{f''(z)}{f'(z)}
    =2a+\sum_{k=1}^n\frac{1}{z-x_k}
       -2\sum_{m=1}^{\mathfrak{r}}
          \frac{1}{z-\zeta_m}.
\]
Near $x_j$, write $f''/f'=(z-x_j)^{-1}+H_j(z)$, where $H_j$ is holomorphic near $x_j$ and
\[
    H_j(x_j)=2a-\sum_{\ell\neq j}\frac{1}{x_\ell-x_j}
                 +2\sum_{m=1}^{\mathfrak{r}}
                    \frac{1}{\zeta_m-x_j}.
\]
Substituting this expansion into the definition of the Schwarzian derivative shows that the coefficient $u_j/2$ of $(z-x_j)^{-1}$ in $S_f$ is $-H_j(x_j)$, see~\eqref{eqn::Schwarz_local_degree_expansion}. Thus $u_j=-2H_j(x_j)$, which gives~\eqref{eqn::quasi_exp_accessory_pole_expansion}.
\medbreak
Finally, let us show $\partial_j \widehat{\LU}_f = u_j$.
By differentiating~\eqref{eqn::quasi_exp_pole_potential} and using $\ud\log|w|=\operatorname{Re}(\ud w/w)$, we have
\begin{align*}
\partial_j\widehat{\LU}_f=&2\sum_{\ell\ne j}\frac{1}{x_\ell-x_j}-4a
    -4\operatorname{Re}\sum_{m=1}^{\mathfrak{r}}\frac{1}{\zeta_m-x_j}\\
&\quad+4\operatorname{Re}\sum_{m=1}^{\mathfrak{r}}
    \left(2a+\sum_{k=1}^n\frac{1}{\zeta_m-x_k}
    -2\sum_{\ell\ne m}\frac{1}{\zeta_m-\zeta_\ell}\right)\partial_j\zeta_m\\
=&2\sum_{\ell\ne j}\frac{1}{x_\ell-x_j}-4a
    -4\operatorname{Re}\sum_{m=1}^{\mathfrak{r}}\frac{1}{\zeta_m-x_j}\tag{due to~\eqref{eqn::quasi_exp_pole_equations}}\\
=&u_j,
\end{align*}
where the last equation is due to~\eqref{eqn::quasi_exp_accessory_pole_expansion} and the fact that the sum over the poles is real by conjugation symmetry.
\end{proof}

\begin{proof}[Proof of Proposition~\ref{prop::quasi_exp_LU_poles}]
Let $P,Q$ vary along the global smooth real branch through $f$ given by Proposition~\ref{prop::quasi_exp_enumeration}, and let $\LU_f$ be defined by~\eqref{eqn::LU_f_def}. The degree $\mathfrak{r}=\deg Q$ is constant along this branch.
\begin{itemize}
\item The function $\LU_f$ is smooth and real-valued on all of $\LX_n$.
Although the individual zeros of $Q$ need not vary smoothly near a multiple zero, the product $\prod_{m=1}^{\mathfrak{r}}P(\zeta_m)$ is a symmetric polynomial in these zeros. It is therefore a polynomial in the coefficients of $P$ and $Q$, and hence varies smoothly along the branch\footnote{This product is the resultant of $P$ and $Q$, up to a multiplicative constant, and can be expressed as a determinant in terms of the coefficients of $P$ and $Q$.}.
Since $P,Q$ are coprime, this product never vanishes. Consequently,
\[
    \sum_{m=1}^{\mathfrak{r}}\log|P(\zeta_m)|
    =\log\left|\prod_{m=1}^{\mathfrak{r}}P(\zeta_m)\right|
\]
is smooth throughout the chamber. Likewise, the sum $\sum_{m=1}^{\mathfrak{r}}\zeta_m$ is a polynomial in the coefficients of $Q$ and is real by conjugation symmetry. Thus every term in~\eqref{eqn::LU_f_def} is smooth and real-valued.

\item Under the assumptions of Lemma~\ref{lem::quasi_exp_Schwarz_accessory_determination}, the two functions $\LU_f$ and $\widehat{\LU}_f$ differ by an additive constant.
Write $\zeta_1,\ldots,\zeta_{\mathfrak{r}}$ for the simple zeros of $Q$, all distinct from the critical points $\{x_1,\ldots,x_n\}$.
The Wronskian identity~\eqref{eqn::quasi_exp_prescribed_Wronskian} gives
\begin{equation*}
    R_{\bs{x}}(\zeta_j)
    :=\prod_{k=1}^n(\zeta_j-x_k)
    =-\frac{P(\zeta_j)Q'(\zeta_j)}{2a}\ne0,
    \qquad j\in\{1,\ldots,\mathfrak{r}\}.
\end{equation*}
Taking absolute values and multiplying over $j$, we obtain
\begin{equation}\label{eqn::quasi_exp_Wronskian_at_poles}
    \prod_{k=1}^n\prod_{m=1}^{\mathfrak{r}}|x_k-\zeta_m|
    =(2a)^{-\mathfrak{r}}\prod_{m=1}^{\mathfrak{r}}|P(\zeta_m)|
      \prod_{1\leq\ell<m\leq \mathfrak{r}}|\zeta_\ell-\zeta_m|^2.
\end{equation}
Here we used $Q'(\zeta_j)=\prod_{m\ne j}(\zeta_j-\zeta_m)$.
Taking logarithms and substituting into~\eqref{eqn::LU_f_def} gives
\begin{equation}\label{eqn::BPZ_solutions_pole_expansion}
    \LU_f(\bs{x})=\widehat{\LU}_f(\bs{x})+4\mathfrak{r}\log(2a).
\end{equation}
Thus Lemma~\ref{lem::quasi_exp_Schwarz_accessory_determination} yields
\[
    \partial_j\LU_f=\partial_j\widehat{\LU}_f=u_j,
    \qquad j\in\{1,\ldots,n\}.
\]

\item The locus where Lemma~\ref{lem::quasi_exp_Schwarz_accessory_determination} applies is dense on every branch; cf.~\cite[Lemma~8.1]{MukhinTarasovVarchenkoSpacesQuasiExponentials}.
Indeed, the identity~\eqref{eqn::quasi_exp_Wronskian_at_poles} shows that the pole assumptions fail precisely when $Q$ has a repeated zero.
This determines a proper algebraic subset of the coefficient space of $Q$.
By the local inverse-function argument in Lemma~\ref{lem::QExprax_enumeration}, the locus of $\bs{x}\in \LX_n$ where $Q$ has repeated zero has empty interior on every branch.
\end{itemize}
Since $\LU_f$ and the accessory parameters are smooth on $\LX_n$, the identity $\partial_j\LU_f=u_j$ extends from this dense locus to the entire chamber $\LX_n$ by continuity.
The accessory equations~\eqref{eqn::quasi_exp_Schwarz_accessory1_x} then imply that $\LU_f$ satisfies~\eqref{eqn::BPZ_chordal_sclimits} with $\lambda=8a^2$.
\end{proof}

\appendix
\section{Rational functions}
\label{sec::rational}
The goal of this appendix is to prove~\eqref{eqn::LSchord_zero}. 
\begin{theorem}\label{thm::BPZ_solutions_zero}
Fix $n\ge 2$ and $\lambda=0$. The number of solutions in $\LSchord{n}^{(0)}$ is given by
\begin{align*}
\#\LSchord{n}^{(0)}=\binom{n}{\lfloor n/2\rfloor}. 
\end{align*}
Moreover, the solutions in $\LSchord{n}^{(0)}$ are all smooth.
\end{theorem}

We first use the Schwarzian derivative to identify solutions of the accessory equations with equivalence classes of rational functions having the prescribed critical points. Lemma~\ref{lem::rational_accessory_counting} then counts these classes and shows that they form global smooth real branches. Finally, Proposition~\ref{prop::rational_LU_poles} constructs an explicit potential along each branch, proving exactness and establishing a bijection with the solutions of the BPZ system at $\lambda=0$. This gives both the stated count and the smoothness of all solutions.

\subsection{Schwarzian derivative of rational functions}
\paragraph*{Rational functions.}
A rational function is a meromorphic function on the Riemann sphere $\widehat{\C}$.
It can be written as $f = P/Q$ for some coprime polynomials $P,Q\in\C[z]$.
We regard two rational functions as equivalent if they differ by post-composition with a M\"obius transformation.
Its degree is
\[
    \deg(f):=\max\{\deg(P),\deg(Q)\}.
\]
Equivalently, $\deg(f)$ is the number of sheets of the branched covering $f\colon\widehat{\C}\to\widehat{\C}$, or the number of preimages of a generic point counted with multiplicity.
The degree is invariant under post-composition with a M\"obius transformation and is therefore well-defined on each equivalence class.

The finite critical points of a rational function are encoded by its Wronskian polynomial
\begin{equation}\label{eqn::rational_Wronskian}
    \mathcal W(f)
    :=\operatorname{Wr}\bigl(Q,P\bigr)
    =P'Q-PQ'.
\end{equation}
Indeed,
\[
    f'(z)=\frac{\mathcal W(f)(z)}{Q(z)^2}.
\]
Changing the representation by post-composing $f$ with a M\"obius transformation multiplies $\operatorname{Wr}(Q,P)$ by a nonzero constant. Hence $\mathcal W(f)$ is well-defined, up to a multiplicative constant, on the equivalence class of $f$. The multiplicity of a finite critical point $\xi\in\C$ is the multiplicity of $\xi$ as a root of $\mathcal W(f)$.

Fix $n\ge 1$ and distinct points $z_1, \ldots, z_n\in\C$. We write $\bs{z}=(z_1, \ldots, z_n)$ and denote by $\Rat(\bs{z})$ the set of equivalence classes of rational functions whose finite critical points are precisely $z_1, \ldots, z_n$, all of which are simple. In terms of the Wronskian~\eqref{eqn::rational_Wronskian}, we have $f\in\Rat(\bs z)$ if and only if
\begin{equation}\label{eqn::rational_prescribed_Wronskian}
    \mathcal W(f)(z)=c\prod_{k=1}^n(z-z_k),\qquad \text{for some }c\in\C\setminus\{0\}.
\end{equation}

\begin{lemma}\label{lem::Schwarz_rational_accessory1}
Fix $n\geq 1$. Fix distinct points $z_1,\ldots,z_n\in\C$ and denote $\bs z=(z_1,\ldots,z_n)$.
Suppose $f\in \Rat(\bs{z})$.
Then the Schwarzian derivative of $f$ can be written as
\begin{equation}\label{eqn::Schwarz_rational_explicit}
    S_f(z) =-\frac32\sum_{k=1}^n\frac{1}{(z-z_k)^2} +\frac{1}{2}\sum_{k=1}^n\frac{u_k}{z-z_k},
\end{equation}
where $(u_1,\ldots,u_n)\in\C^n$ satisfy the following accessory parameter equations:
\begin{align}
    \frac12u_j^2-\sum_{\ell\ne j}\left(\frac{2u_\ell}{z_\ell-z_j}+\frac{6}{(z_\ell-z_j)^2}\right) & =0, \qquad \text{for }j\in\{1,2,\ldots,n\}. \label{eqn::Schwarz_rational_accessory1}
\end{align}
\end{lemma}
\begin{proof}
The argument in the proof of Lemma~\ref{lem::Schwarz_quasi_exp_expansion} applies with $a=0$.
\end{proof}

\begin{lemma}\label{lem::rational_Schwarz_accessory_sum}
Fix $n\ge 1$. Fix distinct points $z_1, \ldots, z_n\in \C$.
Suppose that $(u_1,\ldots,u_n)\in\C^n$ satisfy~\eqref{eqn::Schwarz_rational_accessory1}.
Then 
\begin{align}\label{eqn::Schwarz_rational_accessory2}
\sum_{k=1}^n u_k=0.
\end{align}
\end{lemma}
\begin{proof}
The argument in the proof of Lemma~\ref{lem::quasi_exp_Schwarz_accessory_sum} applies with $a=0$.
\end{proof}

\paragraph*{Degree at infinity.}
Fix $n\ge 1$. Fix distinct points $z_1, \ldots, z_n\in\C$ and denote $\bs{z}=(z_1, \ldots, z_n)$. We further classify the rational functions in $\Rat(\bs{z})$ according to their behavior at infinity.
Let $f\in\Rat(\bs z)$, and choose a rational representative of this equivalence class, still denoted by $f$. Choose a M\"obius transformation $L$ such that $L(f(\infty))=0$. There exist a unique integer $p\ge0$ and a constant $c\in\C\setminus\{0\}$ such that
\begin{equation}\label{eqn::rational_infinity_local_degree}
    (L\circ f)(z)=\frac{c}{z^{p+1}}+O\left(\frac{1}{z^{p+2}}\right),
    \qquad \text{as }z\to\infty.
\end{equation}
We call $p$ the ramification multiplicity of $f$ at infinity. The number $p$ is well-defined on the equivalence class.

We determine the possible values of $p$ directly from the Wronskian. Write $L\circ f=P/Q$ with $P$ and $Q$ coprime, and $\deg(P)<\deg(Q)$. The expansion~\eqref{eqn::rational_infinity_local_degree} gives
\[
    \deg(Q)-\deg(P)=p+1.
\]
Also, Eqs.~\eqref{eqn::rational_Wronskian} and~\eqref{eqn::rational_prescribed_Wronskian} give
\[
    \deg(\mathcal W(f))=\deg(P)+\deg(Q)-1=n.
\]
Consequently,
\[
    \deg(P)=\frac{n-p}{2},
    \qquad
    \deg(Q)=\frac{n+p+2}{2}.
\]
In particular, $0\le p\le n$ and $n-p$ is even. Denote
\[
    \mathsf P_n:=\{p\in\Z:0\le p\le n\text{ and }n-p\text{ is even}\}.
\]
For $p\in\mathsf P_n$, let $\Rat_p(\bs z)$ be the subset of $\Rat(\bs z)$ given by
\begin{equation}\label{eqn::Ratp_def}
    \Rat_p(\bs z)
    :=\{f\in\Rat(\bs z):
          \text{the ramification multiplicity of $f$ at infinity is $p$}\}.
\end{equation}
In particular, we have
\[
    \Rat(\bs z)=\bigsqcup_{p\in\mathsf P_n}\Rat_p(\bs z).
\]

\begin{lemma}\label{lem::Schwarz_rational_accessory3}
Fix $n\geq 1$ and $p\in\mathsf{P}_n$. 
Fix distinct points $z_1,\ldots,z_n\in\C$ and denote $\bs z=(z_1,\ldots,z_n)$.
Suppose $f\in \Rat_p(\bs{z})$.
Then the Schwarzian derivative of $f$ can be written as in~\eqref{eqn::Schwarz_rational_explicit} 
where $(u_1,\ldots,u_n)\in\C^n$ satisfy~\eqref{eqn::Schwarz_rational_accessory1} and 
\begin{align}\label{eqn::Schwarz_rational_accessory3}
\sum_{k=1}^n u_kz_k-3n=-p(p+2).
\end{align}
\end{lemma}
\begin{proof}
Eqs.~\eqref{eqn::Schwarz_rational_explicit} and~\eqref{eqn::Schwarz_rational_accessory1} are proved in Lemma~\ref{lem::Schwarz_rational_accessory1}; it remains to prove~\eqref{eqn::Schwarz_rational_accessory3}.
\begin{itemize}
    \item On the one hand, choose a representative of $f$ and a M\"obius transformation $L$ as in~\eqref{eqn::rational_infinity_local_degree}, and set
\[
    h(\zeta):=(L\circ f)(1/\zeta).
\]
The local degree of $h$ at zero is $p+1$. Hence the local Schwarzian expansion~\eqref{eqn::Schwarz_local_degree_expansion} gives
\[
    S_h(\zeta)
    =\frac{1-(p+1)^2}{2\zeta^2}+O(\zeta^{-1}).
\]
The Schwarzian chain rule~\eqref{eqn::Schwarz_chain_rule} gives
\begin{equation}\label{eqn::Schwarz_rational_expansion_at_infinity}
    S_f(z)
    =-\frac{p(p+2)}{2z^2}+O(z^{-3}),
    \qquad \text{as }z\to\infty.
\end{equation}
\item On the other hand, expanding~\eqref{eqn::Schwarz_rational_explicit} at infinity and using Lemma~\ref{lem::rational_Schwarz_accessory_sum}, we obtain
\[
    S_f(z)
    =\frac{1}{2z^2}
      \left(\sum_{k=1}^n u_kz_k-3n\right)
      +O(z^{-3}).
\]
\end{itemize}
Comparing these two expansions proves~\eqref{eqn::Schwarz_rational_accessory3}.
\end{proof}

\begin{lemma}\label{lem::Schwarz_rational_accessory_refined_reverse}
Fix $n\geq 1$. Fix distinct points $z_1,\ldots,z_n\in\C$ and denote $\bs z=(z_1,\ldots,z_n)$.
Suppose that $(u_1,\ldots,u_n)\in\C^n$ satisfy~\eqref{eqn::Schwarz_rational_accessory1}.
Then there exists a unique equivalence class $f\in\Rat(\bs z)$ whose Schwarzian derivative is given by~\eqref{eqn::Schwarz_rational_explicit}.
\end{lemma}
\begin{proof}
    We follow the proof of Lemma~\ref{lem::Schwarz_quasi_exp_expansion_refined_reverse}. Set
\[
    R(z):=\prod_{k=1}^n(z-z_k),
    \qquad
    A_j:=\sum_{\ell\ne j}\frac{1}{z_\ell-z_j},
    \qquad
    v_j:=u_j-2A_j,
\]
and
\begin{equation}\label{eqn::rational_accessory_ODE_coefficient}
    \psi_{\bs u}(z)
    :=-\frac34\sum_{k=1}^n\frac{1}{(z-z_k)^2}
      +\frac{1}{4}\sum_{k=1}^n\frac{u_k}{z-z_k}.
\end{equation}
The same finite-singularity argument, with $a=0$, shows that the gauge transformation $y=\sqrt R\,w$ turns
\[
    w''(z)+\psi_{\bs u}(z)w(z)=0
\]
into
\begin{equation}\label{eqn::rational_gauged_ODE}
    y''(z)-\frac{R'(z)}{R(z)}y'(z)
    +\frac14\sum_{k=1}^n\frac{v_k}{z-z_k}y(z)=0,
\end{equation}
and every solution of~\eqref{eqn::rational_gauged_ODE} extends to a single-valued entire function, with local exponents $0$ and $2$ at each $z_j$.
\medbreak
Next we analyze the singularity of $\infty$ for~\eqref{eqn::rational_gauged_ODE}. By Lemma~\ref{lem::rational_Schwarz_accessory_sum} and the identity $\sum_{k=1}^n A_k=0$, we have $\sum_{k=1}^n v_k=0$. Therefore,
\[
    -\frac{R'(z)}{R(z)}=-\frac nz+O(z^{-2}),
    \qquad
    \frac14\sum_{k=1}^n\frac{v_k}{z-z_k}=O(z^{-2}),
    \qquad \text{as }z\to\infty.
\]
It follows from~\eqref{eqn::regular_singular_at_infinity} that $\infty$ is a regular singular point of~\eqref{eqn::rational_gauged_ODE}. Since all solutions are single-valued, their monodromy at $\infty$ is trivial, and thus every solution is meromorphic at $\infty$. Being entire, every solution is thus a polynomial.

Choose a polynomial basis $P,Q\in\C[z]$ of solutions of~\eqref{eqn::rational_gauged_ODE}, and set
\[
w_+(z):= \frac{P(z)}{\sqrt{R(z)}}, \qquad w_-(z):= \frac{Q(z)}{\sqrt{R(z)}}.
\]
They are two linearly independent solutions of~\eqref{eqn::ODE_general} with $\psi=\psi_{\bs u}$. Following the same argument as in Lemma~\ref{lem::Schwarz_quasi_exp_expansion_refined_reverse}, the function $f:=P/Q$ belongs to $\Rat(\bs z)$, and Lemma~\ref{lem::Schwarz_ODE} gives $S_f=2\psi_{\bs u}$, which is~\eqref{eqn::Schwarz_rational_explicit}. The same Schwarzian--ODE correspondence shows that any other function with this Schwarzian is the ratio of another basis of solutions, and hence differs from $f$ by a M\"obius transformation. This proves uniqueness of the equivalence class.
\end{proof}

\subsection{Classification of rational functions}
In the following lemma, we collect results on the classification of rational functions from~\cite{MukhinTarasovVarchenkoSchubertCalculus}.
As in the quasi-exponential case, we need both an exact count and smooth real continuation as the critical points vary.

\begin{lemma}\label{lem::rational_accessory_counting}
Fix $n\ge 1$ and $\bs{x}\in\LX_n$.
For $p\in\mathsf{P}_n$, the number of equivalence classes of rational functions in $\Rat_p(\bs{x})$ is
\[\#\Rat_p(\bs{x})=\binom{n}{\frac{1}{2}(n-p)}-\binom{n}{\frac{1}{2}(n-p)-1},\]
where $\binom{n}{-1}:=0$. Consequently, the number of equivalence classes of rational functions in $\Rat(\bs{x})$ is
\[\#\Rat(\bs{x})=\sum_{p\in\mathsf{P}_n}\left[\binom{n}{\frac{1}{2}(n-p)}-\binom{n}{\frac{1}{2}(n-p)-1}\right]=\binom{n}{\lfloor n/2\rfloor}.\]
Moreover, as $\bs x$ varies over $\LX_n$, these equivalence classes of rational functions form $\binom{n}{\lfloor n/2\rfloor}$ global branches. Along each branch, every equivalence class in $\Rat(\bs{x})$ has a representative $f=P/Q$ with real coefficients, and the coefficients of $P$ and $Q$ are smooth functions of $\bs{x}$.
\end{lemma}
\begin{proof}
We follow the proof of Lemma~\ref{lem::QExprax_enumeration} and Proposition~\ref{prop::quasi_exp_enumeration}. This is essentially a restatement of~\cite[Corollary~6.3(ii)]{MukhinTarasovVarchenkoSchubertCalculus} in our notation. Let
\[
    \chi=\left(\frac{n+p}{2},\frac{n-p}{2}\right),
    \qquad \bar\chi=(p,0).
\]
The Wronskian condition identifies $\Rat_p(\bs{x})$ with the $N=2$ Schubert intersection
\[
    \Rat_p(\bs{x})\simeq\Omega_{\Lambda,\bar\chi,\bs{x}}
    :=\Omega_{\bar\chi}(\infty)
      \cap\bigcap_{k=1}^n\Omega_{(1,0)}(x_k),
    \qquad
    \Lambda=((1,0),\ldots,(1,0)),
\]
in the notation of~\cite[Sections~3.1, 3.4, and~4.1]{MukhinTarasovVarchenkoSchubertCalculus}. Since $x_1,\ldots,x_n$ are distinct and real,~\cite[Corollary~6.3(ii)]{MukhinTarasovVarchenkoSchubertCalculus} states that this intersection is transversal and consists of
\[
    \dim\left(V^{\otimes n}\right)^{\mathrm{sing}}_{\chi} = \binom{n}{\frac{1}{2}(n-p)}-\binom{n}{\frac{1}{2}(n-p)-1}
\]
distinct real points.
Here $V=\C^2$ is the standard representation of $\mathfrak{gl}_2$, and $(V^{\otimes n})^{\mathrm{sing}}_{\chi}$ is the subspace of highest-weight vectors of weight $\chi$.
This proves the counting and reality statements. Local smoothness follows from transversality and the inverse function theorem as in Lemma~\ref{lem::QExprax_enumeration}, and the global branches follow from simple connectedness of $\LX_n$ as in Proposition~\ref{prop::quasi_exp_enumeration}.
To choose smooth polynomial representatives, take the unique monic polynomial $P$ of degree $(n-p)/2$ in each space and the unique monic polynomial $Q$ of degree $(n+p+2)/2$ whose coefficient of $z^{(n-p)/2}$ is zero. These are the normalized basis coordinates of the Schubert cell.
\end{proof}

\subsection{Exactness and proof of Theorem~\ref{thm::BPZ_solutions_zero}}
\label{subsec::rational_poles}
In the proof of Proposition~\ref{prop::quasi_exp_accessory_exact}, the $2^n$ distinct solutions of the accessory equations attain the B\'ezout bound. Thus, all these solutions are simple and the Jacobian with respect to the accessory parameters is invertible.
When $\lambda=0$, solutions of the accessory equations need not be simple, so this Jacobian argument is not available in general.
Instead, following Proposition~\ref{prop::quasi_exp_LU_poles}, we construct an explicit potential from the corresponding rational function. Identifying its derivatives with the accessory parameters proves exactness and completes the proof of Theorem~\ref{thm::BPZ_solutions_zero}.

\begin{proposition}\label{prop::rational_LU_poles}
Fix $\lambda=0$ and $n\ge1$. Fix $p\in\mathsf P_n$.
For $\bs{x}=(x_1,\ldots,x_n)\in\LX_n$, let $f=P/Q$ represent a class in $\Rat_p(\bs{x})$, where $P,Q$ are real monic coprime polynomials with
\begin{equation}\label{eqn::rational_accessory_exact_degrees}
    \deg P=\mathfrak{p}:=\frac{n-p}{2},
    \qquad
    \deg Q=\mathfrak{q}:=\frac{n+p+2}{2}.
\end{equation}
Denote the zeros of $Q$ by $\zeta_1,\ldots,\zeta_{\mathfrak{q}}$, counted with multiplicity. Define
\begin{equation}\label{eqn::LU_f_rational_def}
    \LU_f(\bs{x})
    :=-2\sum_{1\le j<k\le n}\log(x_k-x_j)
       +4\sum_{m=1}^{\mathfrak{q}}\log|P(\zeta_m)|.
\end{equation}
Then $\LU_f$ is independent of the normalized representative, it is smooth and real-valued on $\LX_n$, and satisfies
\begin{equation}\label{eqn::rational_potential_gradient}
    \partial_j\LU_f(\bs{x})=u_j(\bs{x}),
    \qquad j\in\{1,\ldots,n\},
\end{equation}
where $(u_1,\ldots,u_n)$ are the accessory parameters of $f$ in~\eqref{eqn::Schwarz_rational_explicit}.
Consequently, $[\LU_f]\in\LSchord{n}^{(0)}$, and the accessory parameters are exact:
\begin{equation}\label{eqn::rational_accessory_exact}
    \partial_j u_i(\bs{x})=\partial_i u_j(\bs{x}),\qquad i\ne j.
\end{equation}
\end{proposition}

As in Section~\ref{subsec::pole}, we first derive the stationary relations and an explicit potential when $Q$ has only simple zeros.
Here, however, we can always arrange this condition locally by replacing $Q$ with $Q+tP$ to split the multiple zeros into simple zeros, without changing the equivalence class or the potential.

\begin{lemma}\label{lem::rational_Schwarz_accessory_determination}
Assume the same setup as in Proposition~\ref{prop::rational_LU_poles}.
Assume in addition that $Q$ has only simple zeros $\zeta_1,\ldots,\zeta_{\mathfrak{q}}$, none of which belongs to $\{x_1,\ldots,x_n\}$.
\begin{itemize}
\item The poles satisfy the stationary relations
\begin{equation}\label{eqn::rational_pole_equations}
    \sum_{k=1}^n\frac{1}{\zeta_j-x_k}
    -2\sum_{m\ne j}\frac{1}{\zeta_j-\zeta_m}=0,
    \qquad \text{for }j\in\{1,\ldots,\mathfrak{q}\}.
\end{equation}
\item The accessory parameters in~\eqref{eqn::Schwarz_rational_explicit} are given by
\begin{equation}\label{eqn::rational_accessory_pole_expansion}
    u_j=2\sum_{\ell\ne j}\frac{1}{x_\ell-x_j}
         -4\sum_{m=1}^{\mathfrak{q}}\frac{1}{\zeta_m-x_j},
    \qquad \text{for }j\in\{1,\ldots,n\}.
\end{equation}
\item Define
\begin{equation}\label{eqn::rational_pole_potential}
    \begin{aligned}
        \widehat{\LU}_f(\bs{x})
        ={}&-2\sum_{1\le j<k\le n}\log(x_k-x_j)
            +4\sum_{k=1}^n\sum_{m=1}^{\mathfrak{q}}\log|x_k-\zeta_m|-8\sum_{1\le\ell<m\le\mathfrak{q}}\log|\zeta_\ell-\zeta_m|.
    \end{aligned}
\end{equation}
Then $\partial_j\widehat{\LU}_f(\bs{x})=u_j$ for all $j\in\{1,\ldots,n\}$.
\end{itemize}
\end{lemma}
\begin{proof}
By~\eqref{eqn::rational_prescribed_Wronskian} and the prescribed degrees of the monic polynomials $P,Q$, we have
\[
    f'(z)=-(p+1)\frac{\prod_{k=1}^n(z-x_k)}{\prod_{m=1}^{\mathfrak{q}}(z-\zeta_m)^2}.
\]
The residue and Schwarzian calculations in the proof of Lemma~\ref{lem::quasi_exp_Schwarz_accessory_determination} apply with $2a\ee^{2az}$ replaced by the nonzero constant $-(p+1)$, giving~\eqref{eqn::rational_pole_equations} and~\eqref{eqn::rational_accessory_pole_expansion}.
The differentiation of $\widehat{\LU}_f$ is the same as in that proof with $a=0$: the stationary relations cancel all terms involving derivatives of the poles, leaving $\partial_j\widehat{\LU}_f=u_j$.
\end{proof}

\begin{proof}[Proof of Proposition~\ref{prop::rational_LU_poles}]
Choose smooth real monic representatives $P,Q$ along the branch through $f$, with degrees given by~\eqref{eqn::rational_accessory_exact_degrees}, and let $\LU_f$ be defined by~\eqref{eqn::LU_f_rational_def}. The degrees $\mathfrak{p},\mathfrak{q}$ are constant along this branch.
\begin{itemize}
\item The function $\LU_f$ is well-defined, smooth, and real-valued on all of $\LX_n$.
Any other real monic pair representing the same class and having the prescribed degrees is of the form $(P,Q+tP)$ for some $t\in\R$.
If $\xi_1,\ldots,\xi_{\mathfrak{p}}$ are the zeros of $P$, counted with multiplicity, then
\[
    \prod_{m=1}^{\mathfrak{q}} P(\zeta_m)
    =(-1)^{\mathfrak{p}\mathfrak{q}}\prod_{\ell=1}^{\mathfrak{p}} Q(\xi_\ell).
\]
Since $(Q+tP)(\xi_\ell)=Q(\xi_\ell)$, this product, and hence~\eqref{eqn::LU_f_rational_def}, is unchanged by the choice of normalized representative.
As in the proof of Proposition~\ref{prop::quasi_exp_LU_poles}, the product is a polynomial in the coefficients of $P,Q$ and never vanishes because $P,Q$ are coprime. Thus
\[
    \sum_{m=1}^{\mathfrak{q}}\log|P(\zeta_m)|
    =\log\left|\prod_{m=1}^{\mathfrak{q}} P(\zeta_m)\right|
\]
is smooth even when $Q$ has repeated zeros, proving the first assertion.

\item Under the assumptions of Lemma~\ref{lem::rational_Schwarz_accessory_determination}, the functions $\LU_f$ and $\widehat{\LU}_f$ differ by an additive constant.
Write $\zeta_1, \ldots, \zeta_{\mathfrak{q}}$ for the simple zeros of $Q$, all distinct from the critical points $x_1, \ldots, x_n$.
The Wronskian identity~\eqref{eqn::rational_prescribed_Wronskian} gives
\begin{equation}\label{eqn::rational_normalized_Wronskian}
    R_{\bs{x}}(\zeta_j)
    :=\prod_{k=1}^n(\zeta_j - x_k)=\frac{P(\zeta_j)Q'(\zeta_j)}{p+1}\ne 0, \quad j\in \{1,\ldots,\mathfrak{q}\}.
\end{equation}
Taking absolute values and multiplying over $j$, we obtain
\begin{equation}\label{eqn::rational_Wronskian_at_poles}
    \prod_{k=1}^n\prod_{m=1}^{\mathfrak{q}}|x_k-\zeta_m|
    =(p+1)^{-\mathfrak{q}}\prod_{m=1}^{\mathfrak{q}}|P(\zeta_m)|
      \prod_{1\leq\ell<m\leq\mathfrak{q}}|\zeta_\ell-\zeta_m|^2.
\end{equation}
The same product calculation as in the proof of Proposition~\ref{prop::quasi_exp_LU_poles} yields
\begin{equation}\label{eqn::BPZ_solutions_zero_pole_expansion}
    \LU_f(\bs{x})=\widehat{\LU}_f(\bs{x})+4\mathfrak{q}\log(p+1).
\end{equation}
Thus Lemma~\ref{lem::rational_Schwarz_accessory_determination} gives 
\[
\partial_j\LU_f=\partial_j\widehat{\LU}_f=u_j, \qquad j\in\{1,\ldots,n\}.
\]

\item For each $\bs{x}\in \LX$, we can always choose an appropriate representative $f=P/Q\in \Rat(\bs{x})$ to ensure that the assumptions of Lemma~\ref{lem::rational_Schwarz_accessory_determination} are satisfied. Indeed, the identity~\eqref{eqn::rational_Wronskian_at_poles} shows that the assumption of Lemma~\ref{lem::rational_Schwarz_accessory_determination} fails if and only if $Q$ has multiple zeros. We can always avoid this by replacing $Q$ with $\widetilde Q=Q+tP$ for some real constant $t$.
\end{itemize}
Thus $\partial_j \LU_f = u_j$ for all $\bs{x}\in \LX_n$ and $j\in\{1,\ldots,n\}$. The accessory equations~\eqref{eqn::Schwarz_rational_accessory1} imply that $\LU_f$ satisfies~\eqref{eqn::BPZ_chordal_sclimits} with $\lambda=0$.
\end{proof}

\begin{proof}[Proof of Theorem~\ref{thm::BPZ_solutions_zero}]
Fix $n\ge2$. Lemmas~\ref{lem::Schwarz_rational_accessory1}, \ref{lem::Schwarz_rational_accessory_refined_reverse}, and~\ref{lem::rational_accessory_counting} give $\binom{n}{\lfloor n/2\rfloor}$ distinct global smooth real-valued branches $\bs{u}=(u_1,\ldots,u_n)$ satisfying~\eqref{eqn::Schwarz_rational_accessory1} with $\bs{z}=\bs{x}$.
For each such branch, Proposition~\ref{prop::rational_LU_poles} gives a smooth real-valued potential $\LU_f$ with $\partial_j\LU_f=u_j$ for all $j\in\{1,\ldots,n\}$.
Substituting these identities into~\eqref{eqn::Schwarz_rational_accessory1}, we obtain
\begin{equation*}
    \frac12(\partial_j\LU_f)^2
    -\sum_{\ell\ne j}\left(
        \frac{2\partial_\ell\LU_f}{x_\ell-x_j}
        +\frac{6}{(x_\ell-x_j)^2}
    \right)=0,
    \qquad j\in\{1,\ldots,n\},
\end{equation*}
which is the BPZ system~\eqref{eqn::BPZ_chordal_sclimits} with $\lambda=0$. Thus each branch gives a solution $[\LU_f]\in\LSchord{n}^{(0)}$.
\medbreak
Conversely, let $[\LU]\in\LSchord{n}^{(0)}$ and set
\[
    \bs{u}:=(\partial_1\LU,\ldots,\partial_n\LU).
\]
Since $\LU$ is $C^1$, the map $\bs{u}$ is continuous, and~\eqref{eqn::BPZ_chordal_sclimits} with $\lambda=0$ implies~\eqref{eqn::Schwarz_rational_accessory1}.
By the Schwarzian correspondence and Lemma~\ref{lem::rational_accessory_counting}, this continuous map lies on one of the above branches.
Proposition~\ref{prop::rational_LU_poles} then gives $\partial_j\LU=u_j=\partial_j\LU_f$.
The connectedness of $\LX_n$ implies that $\LU-\LU_f$ is constant. Different branches have different gradients and hence yield distinct solution classes.
Thus these branches are in one-to-one correspondence with $\LSchord{n}^{(0)}$, and
\[
    \#\LSchord{n}^{(0)}=\binom{n}{\lfloor n/2\rfloor}.
\]
Every solution $\LU$ is smooth, since it differs from the corresponding smooth potential $\LU_f$ by a constant.
\end{proof}

\begin{remark}\label{rem::Rat0_PW}
Fix $n=2N$ even and $p=0$. Fix $\bs{x}=(x_1, \ldots, x_{2N})\in\LX_{2N}$. The number of equivalence classes in $\Rat_0(\bs{x})$ is the $N$-th Catalan number:
\begin{align*}
\#\Rat_0(\bs{x})=\binom{2N}{N}-\binom{2N}{N-1}=\frac{1}{N+1}\binom{2N}{N}. 
\end{align*} 
This conclusion was originally proved in~\cite{EremenkoGabrielov2002Shapiro}. This is equivalent to the renowned Shapiro conjecture in real enumerative geometry, concerning two-dimensional spaces of polynomials.
The connection between $\Rat_0(\bs{x})$ and the BPZ system~\eqref{eqn::BPZ_chordal_sclimits} with $\lambda=0$ was analyzed in~\cite[Section~4]{PeltolaWangSLELDP} and \cite[Section~4]{AlbertsKangMakarovPoleDynamicsMultipleSLE0}.
\end{remark}

\section{Comparison between quantum and semi-classical}
\label{sec::quantum_vs_semiclassical}
In this appendix, we compare the conclusions for the BPZ system~\eqref{eqn::BPZ_chordal} proved in~\cite{ChenHuangWuCommutationChordal} and our conclusions for the semi-classical limits of the BPZ system~\eqref{eqn::BPZ_chordal_sclimits}. 
Fix $n\ge 2$ and $\kappa>0$. For $\lambda\in\R$, define 
\begin{equation}
\QSchord{n}^{(\lambda)}:=\{F\in C^2(\LX_n\to\R): F\text{ satisfies chordal BPZ system~\eqref{eqn::BPZ_chordal}}\}. 
\end{equation}
Recall from~\eqref{eqn::solutions_chordal_def} that 
\begin{equation*}
\LSchord{n}^{(\lambda)}:=\{\LU\in C^{1}(\LX_n\to\R): \LU \text{ satisfies the chordal BPZ system~\eqref{eqn::BPZ_chordal_sclimits}}\}/\sim. 
\end{equation*}

\paragraph*{First Ward identity.}
Solutions in $\QSchord{n}^{(\lambda)}$ are further selected by first Ward identity: for $\nu\in\R$, we consider $\LZ\in\QSchord{n}^{(\lambda)}$ that also satisfies 
\begin{equation}\label{eqn::Ward}
\sum_{j=1}^n\partial_j\LZ=\frac{\nu}{\kappa}\LZ.
\end{equation}
Fix $n\ge 2$ and $\kappa>0$. For $\lambda\in\R$ and $\nu\in\R$, we define 
\begin{align}
	\QSchord{n}^{(\lambda;\nu)}:=\left\{\LZ\in C^2(\LX_n\to \R):
	\begin{array}{l}
		\LZ\text{ satisfies chordal BPZ system~\eqref{eqn::BPZ_chordal}}\\
		\text{and the first Ward identity~\eqref{eqn::Ward}}
	\end{array}\right\}.
\end{align}

Let us consider analogous equation for the semi-classical limits of the BPZ system~\eqref{eqn::BPZ_chordal_sclimits}. 
For $\nu\in\R$, we consider $\LU\in\LSchord{n}^{(\lambda)}$ that also satisfies 
\begin{align}\label{eqn::Ward_sc}
\sum_{j=1}^n \partial_j\LU=-\nu. 
\end{align}
Fix $n\ge 2$. For $\lambda\in\R$ and $\nu\in\R$, we define
\begin{align*}
	\LSchord{n}^{(\lambda;\nu)}:=\left\{\LU\in C^1(\LX_n\to \R):
	\begin{array}{l}
		\LU \text{ satisfies chordal BPZ system~\eqref{eqn::BPZ_chordal_sclimits}}\\
		\text{and the first Ward identity~\eqref{eqn::Ward_sc}}
	\end{array}\right\}/\sim.
\end{align*}


\begin{theorem}\label{thm::lambda_nu_sc}
Fix $n\ge 2$ and $\kappa>0$.
\begin{itemize}
\item When $\lambda>0$, we have
\begin{align}\label{eqn::lambda_nu_sc}
\dim\QSchord{n}^{(\lambda;\nu)}=\#\LSchord{n}^{(\lambda;\nu)}=
\begin{cases}
\displaystyle\binom{n}{\mathfrak{r}}, \quad &\text{if }\nu=(n-2\mathfrak{r})\sqrt{2\lambda}, \quad \mathfrak{r}\in\{0,1,\ldots, n\}; \\
0, &\text{otherwise}.
\end{cases}
\end{align}
Moreover, we have the decomposition
\begin{align}\label{eqn::lambda_nu_decomposition_sc}
\QSchord{n}^{(\lambda)}=\bigoplus_{\mathfrak{r}=0}^n
			\QSchord{n}^{\left(\lambda; \nu_\mathfrak{r} \right)}, \qquad\LSchord{n}^{(\lambda)}=\bigsqcup_{\mathfrak{r}=0}^n 
			\LSchord{n}^{\left(\lambda; \nu_\mathfrak{r} \right)}, \qquad \text{where }\nu_{\mathfrak{r}}:=(n-2\mathfrak{r})\sqrt{2\lambda}.
\end{align}
\item When $\lambda<0$, any solution in $\QSchord{n}^{(\lambda)}$ has at least one zero in $\LX_n$; and $\#\LSchord{n}^{(\lambda)}=0$. 
\end{itemize}
\end{theorem}
\begin{proof}
The conclusions for $\QSchord{n}$ are proved in~\cite[Theorem~1.1]{ChenHuangWuCommutationChordal}. We summarize the proof for $\LSchord{n}$ below.
Suppose $\lambda>0$ and set $a=\sqrt{\lambda/8}$.
By Proposition~\ref{prop::quasi_exp_accessory_exact} and the proof of Theorem~\ref{thm::BPZ_solutions}, solution classes correspond bijectively to the global branches of quasi-exponentials in Proposition~\ref{prop::quasi_exp_enumeration}, with $u_j=\partial_j\LU$.
Lemma~\ref{lem::Schwarz_quasi_exp_expansion_refined} gives
\[
    \sum_{j=1}^n\partial_j\LU=4(2\mathfrak r-n)a
    =-(n-2\mathfrak r)\sqrt{2\lambda}
\]
on the branches with $\deg Q=\mathfrak r$.
The enumeration~\eqref{eqn::lambda_nu_sc} therefore follows from Lemma~\ref{lem::QExprax_enumeration}.
The decomposition~\eqref{eqn::lambda_nu_decomposition_sc} follows from \eqref{eqn::quasi_exp_prescribed_Wronskian_degree}.
When $\lambda<0$, the conclusion follows from Lemma~\ref{lem::BPZ_solutions_neg}.
\end{proof}

The conclusions for $\QSchord{n}$ and for $\LSchord{n}$ in Theorem~\ref{thm::lambda_nu_sc}
have the same admissible Ward parameters and the same binomial multiplicities when $\lambda>0$, but their proofs use different structures.
For $\kappa>0$, we transform the BPZ system into an equivalent first-order system of rational Knizhnik--Zamolodchikov type, and represent the scaled translation operator $\kappa\sum_j\partial_j$ by a matrix. Its eigenvectors are indexed by sign choices $(\sigma_1,\ldots,\sigma_n)\in\{\pm1\}^n$, with eigenvalues $\sqrt{2\lambda}\sum_j\sigma_j$. Counting the choices with exactly $\mathfrak r$ minus signs gives both the Ward parameters $\nu_{\mathfrak r}=(n-2\mathfrak r)\sqrt{2\lambda}$ and their multiplicities $\binom n{\mathfrak r}$; see~\cite[Lemmas~2.3--2.5]{ChenHuangWuCommutationChordal}.
For $\kappa=0$, the admissible Ward parameters follow from the Lax matrix $L$ in Lemma~\ref{lem::quasi_exp_Schwarz_accessory_sum}. The accessory equations imply that its eigenvalues belong to $\{\sqrt{2\lambda},-\sqrt{2\lambda}\}$. Since $\operatorname{tr}L=\sum_j u_j$ and $u_j=\partial_j\LU$, an eigenvalue $\sqrt{2\lambda}$ of algebraic multiplicity $\mathfrak r$ gives $\nu=-\operatorname{tr}L=(n-2\mathfrak r)\sqrt{2\lambda}$. This determines the possible Ward parameters. The number of solution classes at each parameter is then obtained from the Schwarzian correspondence and the enumeration of quasi-exponentials, while the Poisson bracket identity establishes exactness; see Propositions~\ref{prop::quasi_exp_Schwarz}, \ref{prop::quasi_exp_enumeration}, and~\ref{prop::quasi_exp_accessory_exact}.

For $\lambda<0$, both arguments use an energy bound: an integration-by-parts estimate excludes positive BPZ solutions for $\kappa>0$~\cite[Lemma~2.6]{ChenHuangWuCommutationChordal}, whereas the Hamiltonian energy identity~\eqref{eqn::Hamiltonian_energy} excludes global real solutions of the semi-classical system.

\paragraph*{Second Ward identity.}
We next compare the second Ward identities at $\lambda=0$.
For every $\kappa>0$, the only admissible first Ward parameter is $\nu=0$, and $\dim\QSchord{n}^{(0;0)}=\binom n{\lfloor n/2\rfloor}$ by~\cite[Theorem~1.2]{ChenHuangWuCommutationChordal}.
Recall $\mathsf P_n=\{p\in\Z:0\le p\le n,\ n-p\text{ is even}\}$ from Appendix~\ref{sec::rational}, and set
\[
    \mathsf{B}_{n,p}:=\binom{n}{(n-p)/2}-\binom{n}{(n-p)/2-1},
    \qquad p\in\mathsf P_n,
\]
where $\binom{n}{-1}:=0$.
For $\kappa>0$ and $\xi\in\R$, define the space selected by the first two Ward identities by
\[
    \QSchord{n}^{(0;0;\xi)}
    :=\left\{\LZ\in\QSchord{n}^{(0;0)}:
        \sum_{j=1}^n x_j\partial_j\LZ=\frac{\xi}{\kappa}\LZ\right\}.
\]

In the semi-classical setting, every solution at $\lambda=0$ already satisfies the first Ward identity with $\nu=0$, by Lemma~\ref{lem::rational_Schwarz_accessory_sum}. Hence $\LSchord{n}^{(0)}=\LSchord{n}^{(0;0)}$.
The semi-classical second Ward identity is $\sum_j x_j\partial_j\LU=-\xi$.
Accordingly, define
\[
    \LSchord{n}^{(0;0;\xi)}
    :=\left\{\LU\in\LSchord{n}^{(0;0)}:
        \sum_{j=1}^n x_j\partial_j\LU=-\xi\right\}/\sim.
\]


\begin{theorem}\label{thm::zero_Ward_sc}
Fix $n\ge2$ and $\kappa\in(0,8)$ and $\lambda=\nu=0$. Set
\[
    \widehat{\xi}_p:=p\left(p+2-\frac\kappa2\right)
        -n\left(3-\frac\kappa2\right),\qquad \xi_p:=p(p+2)-3n, \qquad p\in\mathsf P_n.
\]
We have
\begin{equation}\label{eqn::Ward_second_count_sc}
    \dim\QSchord{n}^{(0;0;\xi)}=
    \begin{cases}
        \mathsf{B}_{n,p},&\text{if }\xi=\widehat{\xi}_p,\ p\in\mathsf P_n\\
        0,&\text{otherwise}
    \end{cases};\qquad
    \#\LSchord{n}^{(0;0;\xi)}=
    \begin{cases}
        \mathsf{B}_{n,p},&\text{if }\xi=\xi_p,\ p\in\mathsf P_n\\
        0,&\text{otherwise}.
    \end{cases}
\end{equation}
Moreover, we have the decomposition
\begin{equation}\label{eqn::Ward_second_decomposition_sc}
     \QSchord{n}^{(0;0)}=\bigoplus_{p\in\mathsf P_n}
        \QSchord{n}^{(0;0;\widehat{\xi}_p)},\qquad   \LSchord{n}^{(0;0)}=\bigsqcup_{p\in\mathsf P_n}
        \LSchord{n}^{(0;0;\xi_p)}.
\end{equation}
\end{theorem}
\begin{proof}
The conclusions for $\QSchord{n}$ are proved in~\cite[Theorem~1.3]{ChenHuangWuCommutationChordal}. We summarize the proof for $\LSchord{n}$ below.
By Theorem~\ref{thm::BPZ_solutions_zero} and its proof, all solution classes arise from the rational branches in Lemma~\ref{lem::rational_accessory_counting}, with $u_j=\partial_j\LU$.
Lemma~\ref{lem::Schwarz_rational_accessory3} gives
\[
    \sum_{j=1}^n x_j\partial_j\LU=3n-p(p+2)=-\xi_p
\]
on a branch in $\Rat_p(\bs x)$.
The values $\xi_p$ are distinct, so~\eqref{eqn::Ward_second_count_sc} and~\eqref{eqn::Ward_second_decomposition_sc} follow from Lemma~\ref{lem::rational_accessory_counting}.
\end{proof}

\paragraph*{Third Ward identity.}
For an even number $n=2N$ of marked points, the choice $p=0$ in Theorem~\ref{thm::zero_Ward_sc} is special. In this sector, the third Ward identity follows from the first two.

%

\begin{theorem}\label{thm::third_Ward_sc}
Fix $n=2N$ and $\kappa\in(0,8)$ and $\lambda=\nu=0$. Set
\[\widehat{\xi}_0=-n(3-\kappa/2),\qquad \xi_0=-3n.\] 
Then
\[
 \dim\QSchord{n}^{(0;0;\widehat{\xi}_0)}
    =   \#\LSchord{n}^{(0;0;\xi_0)}=\frac{1}{N+1}\binom{2N}{N}=: \Cat_N. 
\]
Moreover, every $\LZ\in\QSchord{n}^{(0;0;\widehat{\xi}_0)}$ satisfies the third Ward identity
\begin{equation}\label{eqn::Ward_third_quantum}
    \sum_{j=1}^n x_j^2\partial_j\LZ
    =-\frac{6-\kappa}{\kappa}\left(\sum_{j=1}^n x_j\right)\LZ;
\end{equation}
and every $\LU\in\LSchord{n}^{(0;0;\xi_0)}$ satisfies the semi-classical third Ward identity
\begin{equation}\label{eqn::Ward_third_sc}
    \sum_{j=1}^n x_j^2\partial_j\LU=6\sum_{j=1}^n x_j.
\end{equation}
\end{theorem}
\begin{proof}
The conclusions for $\QSchord{n}$ are proved in~\cite[Corollary~1.4]{ChenHuangWuCommutationChordal}. The conclusions for $\LSchord{n}$ follow from Theorem~\ref{thm::zero_Ward_sc} and $\mathsf{B}_{2N,0}=\Cat_N$.
The solutions in this sector correspond to the rational branches in $\Rat_0(\bs x)$, with $u_j=\partial_j\LU$.

\begin{itemize}
\item On the one hand, choose the representative in~\eqref{eqn::rational_infinity_local_degree} so that $h(w):=f(1/w)$ is holomorphic near zero with $h'(0)\ne0$, since $p=0$.
The Schwarzian chain rule~\eqref{eqn::Schwarz_chain_rule} gives
\[
    S_f(z)=z^{-4}S_h(1/z)=O(z^{-4}),
    \qquad \text{as }z\to\infty.
\]
\item On the other hand, expanding~\eqref{eqn::Schwarz_rational_explicit} at infinity and using $\sum_j u_j=0$ and $\sum_j x_ju_j=3n$, from Lemma~\ref{lem::rational_Schwarz_accessory_sum} and Theorem~\ref{thm::zero_Ward_sc}, yields
\[
    S_f(z)=\frac{1}{z^3}
    \left(\frac12\sum_{j=1}^n x_j^2u_j-3\sum_{j=1}^n x_j\right)
    +O(z^{-4}).
\]
\end{itemize}
Comparing these two expansions proves~\eqref{eqn::Ward_third_sc}.
\end{proof}

The agreement between quantum dimensions and semi-classical counts does not by itself establish convergence of the corresponding solutions as $\kappa\to0$.
In general, it remains unclear how to choose families of quantum solutions depending on $\kappa$ adapted to the semi-classical branches, so that their logarithmic limits recover those branches, modulo additive constants.
Positivity alone does not resolve this issue: even when a positive basis is available, taking positive linear combinations need not preserve a smooth semi-classical limit.

For example, fix $\lambda>0$, for every $\kappa>0$, define
\begin{align}\label{eqn::LZ_shuffle_def}
\LZ_{\shuffle_n}^{(\pm)}(\bs{x}):=\prod_{1\le j<k\le n}(x_k-x_j)^{2/\kappa}\exp\left(\pm\frac{\sqrt{2\lambda}}{\kappa}\sum_{k=1}^nx_k\right), \qquad \text{for }\bs{x}=(x_1, \ldots, x_n)\in\LX_n.
\end{align}
These are positive solutions of the BPZ system~\eqref{eqn::BPZ_chordal} and form a positive basis of their two-dimensional span.
Comparing with~\eqref{eqn::LU_shuffle_def}, we find
\begin{align}
\LU_{\shuffle_n}^{(\mp)}(\bs{x})=-\lim_{\kappa\to 0}\kappa\log \LZ_{\shuffle_n}^{(\pm)}(\bs{x}).
\end{align}
However, their positive sum $\LZ^{(\kappa)}:=\LZ_{\shuffle_n}^{(+)}+\LZ_{\shuffle_n}^{(-)}$ is also a BPZ solution, and
\[
    -\lim_{\kappa\to0}\kappa\log\LZ^{(\kappa)}(\bs x)
    =-2\sum_{1\le j<k\le n}\log(x_k-x_j)
      -\sqrt{2\lambda}|\sum_{j=1}^n x_j|.
\]
This limit is not differentiable on the hyperplane $\sum_jx_j=0$ in $\LX_n$, and hence does not belong to the $C^1$ solution classes in $\LSchord{n}^{(\lambda)}$.
This illustrates a general feature of logarithmic limits: for finitely many positive solutions $\LZ_{\alpha}^{(\kappa)}$ satisfying $-\kappa\log\LZ_{\alpha}^{(\kappa)}\to\LU_\alpha$ and fixed coefficients $c_\alpha>0$, we have
\[
    -\kappa\log\left(\sum_\alpha c_\alpha\LZ_{\alpha}^{(\kappa)}\right)
    \longrightarrow\min_\alpha\LU_\alpha.
\]
The minimizing branch may change with $\bs x$ and produce a nondifferentiable limit.
How to choose quantum solution families, or a basis, whose logarithmic limits recover all semi-classical branches remains unclear in general.


\end{document}